\documentclass[journal]{IEEEtran}
\usepackage{etex}
\usepackage{times}
\usepackage{latexsym,amsfonts,amsbsy,amssymb}
\usepackage{amsmath}
\usepackage{graphicx}
\usepackage{cite}
\usepackage{algorithm}
\usepackage{algorithmic}
\usepackage{multirow}
\usepackage{extarrows}
\usepackage{tabularx,booktabs}
\usepackage{xcolor}
\usepackage[curve]{xypic}
\usepackage{subfig}
\usepackage{caption}
\usepackage{color}
\usepackage{flushend}
\usepackage{tikz}
\usetikzlibrary{automata,positioning,shapes,arrows}
\newtheorem{theorem}{Theorem}
\newtheorem{lemma}{Lemma}

\newtheorem{definition}{Definition}

\newtheorem{remark}{Remark}

\newtheorem{problem}{Problem}

\usepackage{float}

\begin{document}
\title{\LARGE \bf Learning-based Formal Synthesis of Cooperative Multi-agent Systems\thanks{This work was supported by the National Science Foundation (NSF-CNS-1239222 , NSF-EECS-1253488 and NSF-CNS-1446288)}}
\author{Jin~Dai, %~\IEEEmembership{Student~Member,~IEEE,}
        Alessandro~Benini,
        Hai~Lin, %~\IEEEmembership{Senior~Member,~IEEE}
        Panos~J.~Antsaklis, %~\IEEEmembership{Fellow,~IEEE}
        Matthew~J.~Rutherford,
        and Kimon~P.~Valavanis
\thanks{J. Dai, H. Lin and P.~J.~Antsaklis are with the Department of Electrical Engineering, University of Notre Dame, Notre Dame, IN 46556, USA (e-mail: jdai1@nd.edu; hlin1@nd.edu; pja1@nd.edu).}\\
\thanks{A. Benini, M.~J.~Rutherford and K.~P.~Valavanis are with the Unmanned Systems Research Institute (DU2SRI), Department of Electrical and Computer Engineering, University of Denver, Denver, CO 80208, USA (e-mail: alessandro.benini@du.edu; matthew.rutherford@du.edu; kimon.valavanis@du.edu). }}
\date{}

%%\markboth{06/03/2013}%
%{Shell \MakeLowercase{\textit{et al.}}: Bare Demo of IEEEtran.cls for Journals}

\maketitle
\thispagestyle{empty}
\pagestyle{empty}

\begin{abstract}
We propose a formal design framework for synthesizing coordination and control policies for cooperative multi-agent systems to accomplish a global mission. The global performance requirements are specified as regular languages while dynamics of each agent as well as the shared environment are characterized by finite automata, upon on which a formal design approach is carried out via divide-and-conquer. Specifically, the global mission is decomposed into local tasks; and local mission supervisors are designed to accomplish these local tasks while maintaining the multi-agent performance by integrating supervisor synthesis with compositional verification techniques; finally, motion plans are automatically synthesized based on the obtained mission plans. We present three modifications of the $L^*$ learning algorithm such that they are adapted for the synthesis of the local mission supervisors, the compositional verification and the synthesis of local motion plans, to guarantee that the collective behavior of the agents will ensure the satisfaction of the global specification. Furthermore, the effectiveness of the proposed framework is demonstrated by a detailed experimental study based on the implementation of a multi-robot coordination scenario. The proposed hardware-software architecture, with each robot's communication and localization capabilities, is exploited to examine the automatic supervisor synthesis with inter-robot communication.
\end{abstract}

%\IEEEpeerreviewmaketitle

\section{Introduction}
\IEEEPARstart{C}{ooperative} multi-agent systems that consist of a number of agents collaborating through distributed physical interaction and wireless communication to fulfill certain performance objectives have emerged as a hot research topic in the past two decades, due to their wide applications in both academia and industry, ranging from power grids, transportation systems, computer networks to robotic teams. Numerous distributed coordination and control problems have been extensively studied, see, e.g., \cite{chen,klo2,reyes,osm2,pan,val} and the references therein.

One of the essential problems in cooperative multi-agent systems is how to design control policies for each agent and coordination strategies among them such that certain desirable global specifications can be fulfilled.
%To pursue satisfaction of desired performance requirements, coordination methods for cooperative multi-agent systems can generally be divided into two categories: {\it bottom-up} and {\it top-down} approaches. Bottom-up approaches (see e.g. \cite{pan,guo}) design local control rules and inter-agent coordination mechanisms to fulfill each agent's individual tasks, while top-down approaches (see e.g. \cite{klo2,kar}) decompose a global task into a series of local tasks for each agent in such a way that accomplishment of the local tasks assures the satisfaction of the global specification.
We are therefore motivated to derive a formal synthesis framework to assure the accomplishment of high-level missions for a team of cooperative agents. In particular, we consider multiple intelligent agents that can move around in a partitioned and uncertain environment while satisfying a team mission specified by a regular language \cite{cas}. We assume that the mission and motion capabilities of each agent can be respectively modeled as two finite automata \cite{baier}. A learning-based formal synthesis approach is performed to coordinate the cooperative multi-agent systems as follows. First, the global mission is properly decomposed into local missions that can be executed by the corresponding agents. Towards this end, we present a learning-based algorithm for synthesizing local mission supervisors to fulfill local missions. Secondly, after deducing the local tasks, a learning-based algorithm is provided to compositionally justify whether or not the collective behavior of all agents satisfies the global mission \cite{cob2}. Finally, a motion plan \cite{dasilva} corresponding to the local mission is generated based on a nominal environment model. In case the motion plans fail to be feasible due to the uncertainties of the real environment, motion replanning is triggered and the global mission will be achieved iteratively. The contributions of our work can be summarized as:

{\it 1)} Compared to \cite{chen,dasilva}, we account for the controllability of the mission events and derive local supervisors that restrict the behavior of each agent properly. Furthermore, we modify the $L^*$ learning algorithm \cite{ang} and apply the modified algorithm to the synthesis of the local supervisors to actively learn a manageable local task even if the agent's model is not known {\it a priori}, rendering our approach the capability of handling system uncertainties.

{\it 2)} We adopt compositional verification to examine the correctness of the joint effort of all the agents and mitigate the computational complexity by introducing an assume-guarantee scheme \cite{cob2}. Another modification of the $L^*$ learning algorithm is developed to generate appropriate assumptions for each supervised agent automatically.

{\it 3)} In addition to our previous work \cite{jin1,jin2}, we further associate synthesized mission plans with appropriate motion plans by presenting a third modification of the $L^*$ learning algorithm. Moreover, replanning schemes of the motion plan in the practical environment are also studied.

{\it 4)} To examine the effectiveness of the design framework, we apply the theoretical results to a practical multi-robot coordination demonstration. It turns out in the experiment that the robots can cooperatively satisfy a request-response and coordination task even the environment possesses uncertainties.

%Compared to our previous work on learning-based supervisor design problems \cite{jin}, we deal with synthesis problems of local supervisors corresponding to each distributed agent rather than decentralized supervisory control problems, which aim at controlling a monolithic plant with multiple supervisors. Preliminary versions of the current work addressing the formal synthesis of cooperative multi-agent systems were reported in \cite{jin1,jin2}. Compared to \cite{jin2}, this work provides the detailed control and verification algorithms for cooperative multi-agent systems, along with theoretical proofs for the correctness and convergence, in addition to the demonstration of the synthesis framework on practical robotic team cooperation tasks.

The remainder of this paper is organized as follows. In Section~\ref{sec:related work}, we briefly review prior work on multi-agent coordination and application of formal methods. In Section~\ref{sec:preliminaries}, we give a brief description of the supervisory control theory and compositional verification within the multi-agent system formalism as well as the $L^*$ learning algorithm. Based on the preliminaries, we formulate the multi-agent coordination and control co-design problem and propose an automatic synthesis framework to synthesize appropriate local mission and motion plans for each agent in Section~\ref{sec:formulation}. Two modified $L^*$ algorithm are presented in Section~\ref{sec:synthesis} to generate appropriate local mission plans via supervisor synthesis and compositional verification. In Section~\ref{sec:verification}, we incorporate the model of the nominal environment with the mission plans to synthesize corresponding motion plans by utilizing a third modified $L^*$ algorithm, and online replanning strategies are proposed to address environment uncertainties. We apply the proposed framework to a multi-robot system in order to examine its effectiveness in Section~\ref{sec:CoSMoP}, while Section~\ref{sec:experimental_setup} demonstrates the experimental results for the proposed robotic coordination scenario by deploying three robots that are equipped with dedicated wireless communication capabilities, localization cameras, and the proposed supervisors that are implemented as MATLAB StateFlow machines. Concluding remarks and future work are discussed in Section~\ref{sec:conclusion}.

\section{Related Work}\label{sec:related work}
%Coordination and planning of multi-agent systems have been fundamental research problems of many recent studies. Considerable amount of existing research efforts has been devoted to behavior-based approach \cite{ark}, in which coordination of multi-agent systems can be accomplished by taking advantage of pre-defined behaviors or by exploiting learning algorithms from distributed artificial intelligence \cite{ferb,la}. Despite the capability of adapting to various team tasks, much of the behavior-based work shows empirical features, resulting in a trial-and-error design process and therefore lacking guarantees of the performance. Although recent study \cite{lyons} has taken performance verification of behavior-based agents into consideration, it mainly focuses on the mission planning problem of one single agent.

The past two decades have witnessed the development of performance-guaranteed design of cooperative multi-agent systems, and blooming contributions have been made to addressing various distributed coordination and control purposes, such as consensus \cite{osm2}, flocking \cite{osm1}, rendezvous \cite{dim}, formation control \cite{reyes} and cooperative learning \cite{la}. The achievement of the global coordination goals is enforced by Lyapunov stability \cite{osm1}, barrier certificates \cite{pan} and game theory \cite{mard}. Despite fulfillment of steady-state performance objectives, satisfaction of more complicated and temporal specifications is not investigated in these previous works.

The desire to improve the expressiveness of tasks drives our interest towards {\it formal} specifications, including regular languages, linear temporal logic (LTL) and computation tree logic (CTL) formulas \cite{baier}. A recent trend of ``symbolic planning" \cite{belta,fai} has established the standard procedure for synthesizing correct-by-construction controllers for formal specifications by combining abstraction-based approaches \cite{klo,tab} with formal methods \cite{ram,ram2,gaz2}. However, most of the existing works usually assume full knowledge of the environment. To cope with the uncertain environments, Chen et al. \cite{chen2} designed a control policy to achieve a surveillance mission by combining formal synthesis methods with automata learning of the environment, whereas Lahijanian et al. \cite{lahi} employed a multilayered synergistic planner to generate trajectories that partially satisfy an LTL specification in uncertain environments; note that both of the aforementioned works focused on single-agent case.
%In addition to environment uncertainties, Pham et al. \cite{pham} studied the uncertainties from the specification languages and considered the transparency of control specifications; however, possible uncertainties arising in the environment were neglected.

Although control of a single agent to attain formal specifications is tractable, extensions to the multi-agent case are non-trivial. Filippidis et al. \cite{fil} proposed a decentralized control scheme to satisfy local LTL specifications with communication constraints among the agents; and this idea was developed by Guo and Dimarogonas \cite{guo} to derive a partially decentralized coordination scheme that decomposed the team into clusters of dependent agents to fulfill local LTL specifications. To cope with computational complexity issues, the results were further extended by Tumova et al. \cite{tum} by employing receding horizon planning techniques. Nevertheless, specifying each agent's individual task is difficult to scale up and usually becomes burdensome for human operators in application.

Our work is most closely related to coordination approaches presented in \cite{chen}, \cite{kar} and \cite{seow}. Facing a global LTL specification, Kloetzer and Belta \cite{klo2} derived a centralized solution and assigned individual tasks for each agent accordingly. Nonetheless, deployment of such approaches suffered from the ``state explosion" due to its monolithic design pattern. ``Trace-closed" regular language specifications were studied in \cite{chen} for automatic deployment of robotic teams. Karimadini and Lin \cite{kar} presented necessary and sufficient conditions under which the global tasks can be retrieved by local ones, while Partovi and Lin \cite{ali} investigated this problem for CTL specifications. Brandin et al. \cite{brandin} studied the coordinated planning of multi-agent systems by using an incremental verification approach, whereas Seow et al. \cite{seow} investigated the coordination of multi-agent systems from a supervisory control perspective with additional ``coordination modules". These results focused on task allocation and/or design of control policies among multiple agents with known external environment, while in our proposed work, coordination and control co-design in the presence of both system and environment uncertainties are investigated.

\section{Preliminaries}\label{sec:preliminaries}

This section presents preliminaries in automata theory, supervisory control, compositional verification, and a brief introduction of the $L^*$ learning algorithm \cite{ang}. Throughout the rest of this paper, we use the following notations. Let $\mathbb{N}$ denote the set of non-negative integers. For a finite set $\Sigma$, let $2^\Sigma$ and $|\Sigma|$ denote its power set (set of all subsets) and cardinality, respectively. For two sets $\Sigma_1$ and $\Sigma_2$, $\Sigma_1-\Sigma_2$ denotes the set-theoretic difference of $\Sigma_1$ and $\Sigma_2$; $\Sigma_1\Delta\Sigma_2$ denotes the symmetric difference of of $\Sigma_1$ and $\Sigma_2$, i.e., $\Sigma_1\Delta\Sigma_2=(\Sigma_1-\Sigma_2)\cup(\Sigma_2-\Sigma_1)$; and $\Sigma_1\Sigma_2=\{\sigma_1\sigma_2|\sigma_1\in\Sigma_1, \sigma_2\in\Sigma_2\}$, where $\sigma_1\sigma_2$ stands for the concatenation of two elements $\sigma_1$ and $\sigma_2$.

\subsection{Automata and Regular Language Models}
For a given finite set (alphabet) of {\it events} $\Sigma$, a finite sequence $w$ composing of elements in $\Sigma$, i.e., $w=\sigma_0\sigma_1\ldots \sigma_m$, is called a {\it word} over $\Sigma$. We use $\Sigma^*$ to denote the {\it Kleene-closure} of $\Sigma$ \cite{baier}, including the empty word $\epsilon$. The length of a word $w\in \Sigma^*$ is denoted by $|w|$. A subset of $\Sigma^*$ is called a {\it language} over $\Sigma$. The {\it prefix-closure} of a language $L\subseteq \Sigma^*$, denoted as $\overline{L}$, is the set of all {\it prefixes} of words in $L$, i.e., $\overline{L}=\{s\in\Sigma^*|(\exists t\in\Sigma^*)[st\in L]\}$. $L$ is said to be {\it prefix-closed} if $\overline{L}=L$. For two languages $L_1,L_2\subseteq \Sigma^*$, the {\it quotient} is the collection of prefixes of words in $L_1$ with a suffix that belongs to $L_2$, i.e., $L_1/L_2=\{s\in\Sigma^*|(\exists t)[(t\in L_2 )\land (st \in L_1)]\}$.
%We also define that $L_1\setminus L_2=\{t\in\Sigma^*\vert (\exists s\in L_2)[st\in L_1]\}$.

To recognize languages over $\Sigma$, the definition of deterministic finite automaton \cite{baier} is recalled as follows.

\begin{definition}[Deterministic Finite Automaton]\label{dfa}
A deterministic finite automaton (DFA) is a 5-tuple
$$G=(Q,\Sigma,q_0,\delta,Q_m),$$
where $Q$ is a finite set of states, $\Sigma$ is a finite set of events, $q_0\in Q$ is an initial state, $\delta: Q\times \Sigma \to Q$ is a partial transition function and $Q_m\subseteq Q$ is the set of the marked (accepting) states.
\end{definition}

The transition function $\delta$ can be extended to $\delta: Q\times \Sigma^* \to Q$ in the usual manner \cite{cas}. We use the notation $\delta(q,s)!$ to denote that the transition $\delta(q,s)$ is defined. A DFA $G$ is said to be {\it accessible} if $(\forall q\in Q)(\exists s\in\Sigma^*)[\delta(q_0,s)=q]$; $G$ is said to be {\it coaccessible} if $(\forall q\in Q)(\exists s\in\Sigma^*)[\delta(q,s)\in Q_m]$; $G$ is said to be {\it trim} if it is both accessible and coaccessible \cite{cas}.

A {\it run} of a DFA $G$ on a finite word $w=\sigma_0\sigma_1\ldots \sigma_m\in\Sigma^*$, is a sequence of states $Run(w)=q_0q_1\ldots q_{m+1}\in Q^*$ satisfying $q_{i+1}=\delta(q_i,\sigma_i) (i=0,1,\ldots,m)$. The language generated by a DFA $G$ is given by $L(G)=\left\{s\in \Sigma^* \vert \delta(q_0,s)!\right\}$, and the language accepted (recognized) by $G$ is given by $L_m(G)=\left\{s\in \Sigma^* \vert s\in L(G), \delta(q_0,s)\in Q_m\right\}$. A language $L\subseteq\Sigma^*$ is said to be {\it regular} if there exists a DFA $G$ such that $L(G)=\overline{L}$ and $L_m(G)=L$. We will focus our study on regular languages throughout the rest of this paper.

%The {\it prefix-closure} of a regular language $L\subseteq \Sigma^*$, denoted as $\overline{L}$, is the set of all {\it prefixes} of words in $L$, i.e., $\overline{L}=\{s\in\Sigma^*|\exists t\in\Sigma^*: st\in L\}$, where $st$ denotes the concatenation of two words $s$ and $t$. $L$ is said to be {\it prefix-closed} if $\overline{L}=L$. For two languages $L_1,L_2\subseteq \Sigma^*$, the {\it right quotient} (or simply quotient) is the collection of prefixes of words in $L_1$ with a suffix that belongs to $L_2$, i.e., $L_1/L_2=\{s\in\Sigma^*|(\exists t)[(t\in L_2 )\land (st \in L_1)]\}$. We also define that $L_1\setminus L_2=\{t\in\Sigma^*\vert (\exists s\in L_2)[st\in L_1]\}$.

The {\it complement language} of a regular language $L$ over $\Sigma$ is defined as $coL=\Sigma^*-L$. Let $G$ be a DFA accepting $L$, the {\it complete} DFA of $G$, denoted as $\tilde{G}$, is defined as follows.

\begin{definition}[Complete DFA]\label{complement DFA}
Let $q_e$ be an ``error state". Given a DFA $G=(Q,\Sigma, q_0,\delta,Q_m)$, the complete model of $G$ is defined as a DFA $\tilde{G}=(\tilde {Q},\Sigma, q_0,\tilde{\delta}, \{q_e\})$, where $\tilde{Q}=Q\cup \{q_e\}$, and
$$
\forall \tilde{q}\in \tilde {Q}, \sigma\in \Sigma, \tilde {\delta}(\tilde{q},\sigma)=
\begin{cases}
\delta(\tilde{q},\sigma), & \mbox{if }\tilde {q}\in Q \land \delta(\tilde {q},\sigma)!, \\
q_e, & \mbox{if } \tilde{q} = q_e \lor \delta(\tilde {q},\sigma)=\emptyset.
\end{cases}
$$
\end{definition}

It can be verified that $L_m(\tilde{G})=\Sigma^*$. The {\it complement DFA} $coG$ of $G$ that accepts $coL$, is formed by swapping the marked states of $G$ with its non-marked states and vice versa, i.e., $coG=(\tilde Q,\Sigma, q_0,\tilde{\delta}, \tilde Q-Q_m)$.

We conclude the review of automata theory by investigating the concurrent operation of two or more DFAs.

\begin{definition}[Parallel Composition]\label{parallel}
\cite{cob2,cas} Given two finite automata $G_i=(Q_i,\Sigma_i,q_{i,0},\delta_i,Q_{i,m}) (i=1,2)$, the parallel composition of $G_1$ and $G_2$, denoted by $G_1\vert\vert G_2$, is defined as the finite automata
$G_1\vert\vert G_2=(Q,\Sigma, q_0,\delta,Q_m)$, where $Q=Q_1\times Q_2$ is the set of states, $\Sigma=\Sigma_1\cup\Sigma_2$ is the set of events, $q_0=(q_{1,0}, q_{2,0})$ is the initial state, $Q_m=Q_{1,m}\times Q_{2,m}$ is the set of marked states, and the $\delta : Q\times \Sigma \to Q$ is the transition function that is given by:
\begin{equation*}
\begin{split}
&\forall q=(q_1,q_2)\in Q, \sigma\in \Sigma, \\
&\delta(q,\sigma)=
\begin{cases}
(\delta_1(q_1,\sigma), \delta_2(q_2,\sigma)), & \mbox{if }\delta_1(q_1,\sigma)!, \delta_2(q_2,\sigma)!,\\
(\delta_1(q_1,\sigma), q_2), & \mbox{if }\delta_1(q_1,\sigma)!, \sigma\not\in \Sigma_2,\\
(q_1, \delta_2(q_2,\sigma)), & \mbox{if }\delta_2(q_2,\sigma)!, \sigma\not\in \Sigma_1,\\
\varnothing, &\mbox{otherwise,}
\end{cases}
\end{split}
\end{equation*}
\end{definition}

\begin{remark}
%In the parallel composition of two automata, the two automata are ``synchronized" on the common events in $\Sigma_1\cap\Sigma_2$. The other events in $\Sigma_1\cup\Sigma_2$ are not subject to such a constraint and can be executed whenever possible.
It is worth pointing out that parallel composition of more than two finite automata can be defined recursively based on the associativity of parallel composition.
\end{remark}

The connection between the behavior of the composed DFA and a local DFA is captured by the ``natural projection" from the global event set to the local one. For non-empty event sets $\Sigma$ and $\Sigma'$ satisfying $\Sigma'\subseteq\Sigma$, the natural projection $P_{\Sigma'}: \Sigma^*\to \Sigma'^*$ is inductively defined as: $P_{\Sigma'}(\epsilon)=\epsilon$; $P_{\Sigma'}(\sigma)=\sigma$ if $\sigma\in\Sigma'$; $P_{\Sigma'}(s\sigma)=P_{\Sigma'}(s)P_{\Sigma'}(\sigma)$, $\forall s\in \Sigma^*, \sigma\in \Sigma$. The set-valued inverse projection $P_{\Sigma'}^{-1}:\Sigma'^* \to 2^{\Sigma^*}$ of $P_{\Sigma'}$ is defined as $P_{\Sigma'}^{-1}(s)=\left\{t\in \Sigma^*: P_{\Sigma'}(t)=P_{\Sigma'}(s)\right\}$.
%$P_i(\epsilon)=\epsilon$; $P_i(\sigma)=\sigma$ if $\sigma\in\Sigma_i$;
%$P_i(s\sigma)=P_i(s)P_i(\sigma)$, $\forall s\in \Sigma^*, \sigma\in \Sigma$. The inverse %projection $P_i^{-1}:2^{\Sigma_i^*} \to 2^{\Sigma^*}$ is then defined as $P_i^{-1}(s)=\left\{t\in %\Sigma^*: P(t)=P(s)\right\}$.
%\begin{definition}[Natural Projection]
%\cite{cas}
%$$$
%$$
%\forall s\in \Sigma^*, \sigma\in \Sigma, P(s\sigma)=
%\begin{cases}
%P(s),  & \mbox{if } \sigma\notin \Sigma' \\
%P(s)\sigma, &\mbox{otherwise}
%\end{cases}
%$$
%\end{definition}
%

Let $I=\{1,2,\ldots,n\}$ be an index set. For local event sets $\Sigma_i (i\in I)$ and the global event set $\Sigma=\bigcup_{i\in I} \Sigma_i$, let $P_i$ denote the natural projection from $\Sigma^*$ to $\Sigma_i^*$. The {\it synchronous product} of a finite set of regular languages $L_i\subseteq \Sigma_i^* (i\in I)$, denoted by $\vert\vert_{i\in I} L_i$, is defined as follows.

\begin{definition}[Synchronous Product]\label{product}
\cite{will} For a finite set of regular languages $L_i\subseteq \Sigma_i^* (i=1,2,\ldots,n)$,
\begin{equation}
\vert\vert_{i\in I} L_i=\bigcap_{i\in I} P_i^{-1}(L_i).
\end{equation}
\end{definition}

It can be shown that for $G=\vert\vert_{i=1}^n G_i$,
$L(G)=\vert\vert_{i=1}^n L(G_i)$.

The notion of language separability plays an essential role in the rest of the paper and is defined formally as follows.

\begin{definition}[Separable Languages]\label{sep}
\cite{will} For the local event sets $\Sigma_i (i=1,2,\ldots,n)$ and the global event set $\Sigma=\bigcup_{i=1}^n \Sigma_i$, a language $L\subseteq \Sigma^*$ is said to be {\it separable} with respect to $\{\Sigma_i\}_{i=1}^n$ if there exists a set of local languages $L_i\subseteq \Sigma_i^*$ for each $i\le n$ such that $L=\vert\vert_{i=1}^n L_i$.
\end{definition}

It has also been shown in \cite{will} that $L\subseteq \Sigma^*$ is separable with respect to $\{\Sigma_i\}_{i=1}^n$ if and only if $L=\vert\vert_{i=1}^n P_i(L)=\bigcap_{i=1}^n P_i^{-1}[P_i(L)]$.

\begin{remark}
For the local event sets $\Sigma_i (i\in I)$ and the global event set $\Sigma=\bigcup_{i\in I} \Sigma_i$, let $D(\Sigma)=\{(\sigma_1,\sigma_2)\in \Sigma\times\Sigma| \exists i\in I, \sigma_1,\sigma_2\in\Sigma_i\}$. The independence relation is defined as $I(\Sigma)=\Sigma\times\Sigma -D(\Sigma)$ \cite{lss}. It follows from \cite{lss} that a language $L\subseteq L(G)$ possesses a non-empty separable sublanguage if and only if $I(\Sigma)$ is transitive.
\end{remark}

\subsection{Supervisory Control Theory}
Given a system modeled by a DFA $G=(Q,\Sigma,q_0,\delta,Q_m)$, let $\Sigma$ be partitioned into the controllable event set $\Sigma_c$ and the uncontrollable event set $\Sigma_{uc}$, i.e., $\Sigma=\Sigma_{c} \dot \cup \Sigma_{uc}$. A {\it supervisor} \cite{ram,ram2} $S$ is another DFA over $\Sigma$ with the constraint that each state of $S$ is marked. $S$ operates in parallel with $G$ and modifies the behavior of $G$ by disabling certain controllable events. The behavior of $G$ under control of $S$ is denoted as $L(S\vert\vert G)$. Given a prefix-closed and non-empty specification language $L=\overline{L}\subseteq L(G)$, there exists a local supervisor $S$ such that $L(S\vert\vert G)=L$ if and only if $L$ is {\it controllable} with respect to $G$ and $\Sigma_{uc}$.

%To pursue supervisory control \cite{ram}\cite{ram2} of a DFA $G$ over $\Sigma$, $\Sigma_i$ is partitioned into the set of (local) controllable events and the set of uncontrollable events, i.e., $\Sigma_i=\Sigma_{i,uc} \dot \cup \Sigma_{i,c}$. In practice, an uncontrollable event may represent an action that cannot be prohibited or controlled in advance, for example a communication message from another agent. A local supervisor $S_i$ associated with $G_i$ is another automaton that operates in parallel with $G_i$, and the local controlled behavior can then be modeled as $L(S_i\vert\vert G_i)$.

\begin{definition}[Controllable Languages]\label{controllability}
\cite{cas,ram} A language $L\subseteq L(G)$ is said to be controllable with respect to $G$ and $\Sigma_{uc}$ (or controllable for short) if
%$(\forall s\in\overline{L})$$(\forall \sigma\in\Sigma_{uc})$$[s\sigma\in L(G)\Rightarrow s\sigma\in L]$; or equivalently,
$\overline{L}\Sigma_{uc}\cap L(G) \subseteq \overline{L}$.
\end{definition}

If the controllability of $L$ fails, a supervisor is synthesized for the supremal controllable (also prefix-closed) sublanguage of $L$, denoted as $\sup C(L)$.
%
%\begin{remark}
%Throughout the rest of this paper, we assume that prior knowledge of $G_i$, $i\in I$ is not necessarily available, which is quite reasonable due to the uncertainty or scale of $G_i$. Nevertheless, we assume that $\Sigma_i$, $i\in I$ is accessible information for coordination and control purposes.
%\end{remark}

\subsection{Compositional Verification}
We employ a compositional verification procedure to evaluate the coordinated behavior of the multi-agent system. First, we present the concept of ``property satisfaction" in the context of formal verification.

\begin{definition}[Satisfaction Relation]\label{satisfaction}
Given a system modeled by a DFA $M=(Q,\Sigma_M, q_0,\delta,Q_m)$ and a regular language (property) that is accepted by a DFA $P=(Q_P,\Sigma_P, q_{0,P}, \delta_P, Q_{m,P})$ with $\Sigma_P \subseteq \Sigma_M$, the system $M$ is said to satisfy $P$, written as $M \models P$, if and only if $\forall t\in L_m(M): P_P(t)\in L_m(P)$, where $P_P$ denotes the natural projection from $\Sigma_M$ to $\Sigma_P$.
\end{definition}

%In case where $\Sigma_P=\Sigma_M$, $M\models P$ is equivalent to $L_m(M)\subseteq L_m(P)$, indicating that the satisfaction relation reduces to a language inclusion relation.
Note that when $P$ is represented by a prefix-closed language, the marked languages considered in Definition \ref{satisfaction} can be altered by generated languages.

To mitigate the ``state explosion" issue arising in verification for systems that consist of multiple components, the assume-guarantee reasoning scheme \cite{cob2} is employed at this point. In this paradigm, a formula to be checked is a triple $\langle A \rangle M \langle P \rangle $, where $M$ is a system model, $P$ is a property and $A$ is an assumption about $M$'s environment, all of which are represented by a corresponding DFA. The formula holds if whenever $M$ is part of a system satisfying $A$, then the system must guarantee the property $P$, i.e., $\forall E$, $E\vert\vert M \models A$ implies that $E\vert\vert M\models P$ \cite{cob2}. The following theorem provides a DFA characterization of the assume-guarantee satisfaction.

\begin{theorem}\label{assume-guarantee}
\cite{cob2} $\langle A \rangle M \langle P \rangle $ holds if and only if $q_e$ is unreachable in $A\vert\vert M \vert\vert \tilde{P}$ ($\tilde{P}$ is the complete DFA of $P$).
\end{theorem}

The assume-guarantee scheme summons a series of symmetric and/or asymmetric proof rules to ease the reasoning when the system $M$ consists of multiple components, i.e., $M=\vert\vert_{i=1}^n M_i$. We apply the symmetric proof rule, namely SYM-N \cite{cob2}, for the compositional verification, where $A_i$ is an assumption about $M_i$'s environment and $coA_i$ is the complement DFA of $A_i$.

\begin{center}
\begin{tabular}{ll}
1 & $\langle A_1 \rangle M_1 \langle P \rangle $ \\
$\cdots$ & \\
$n$ & $\langle A_n\rangle M_n \langle P \rangle$ \\
$n+1$ & $L_m(coA_1\vert\vert coA_2 \vert\vert \cdots \vert\vert coA_n)\subseteq L_m(P)$ \\
\hline
 & $\langle true \rangle (M_1\vert\vert M_2 \vert\vert \cdots \vert\vert M_n)\langle P \rangle$
\end{tabular}
\end{center}
%where $coA_i$ is the complement DFA of $A_i$.

A central notion of the compositional verification is the weakest assumption, which is defined formally as follows.

\begin{definition}[Weakest Assumption for $\Sigma_{i,IF}$]\label{weakest assumption}
(Adapted from \cite{cob2}) Let $M=\vert\vert_{i=1}^n M_i$ be a system with local alphabets $\Sigma_i (i=1,2,\ldots,n)$, $P$ be a DFA representation of a property and $\Sigma_{i,IF}$ be a specified interface alphabet of $M_i$ to its environment. The weakest assumption $A_i^{w,\Sigma_{i,IF}}$ of $M_i$ over the alphabet $\Sigma_{i,IF}$ for property $P$ is a DFA such that: 1) $\Sigma(A_i^{w,\Sigma_{i,IF}})=\Sigma_{i,IF}$; 2) for any $M_{-i}=\vert\vert_{j\in I,j\ne i} M_j$, $\langle true \rangle M_i\vert\vert (P_{i,IF}(M_{-i}))\langle P \rangle$ if and only if $\langle true \rangle M_{-i} \langle A_i^{w,\Sigma_{i,IF}} \rangle$, where $P_{i,IF}$ is the natural projection from $\bigcup_{i=1}^n \Sigma_i$ to $\Sigma_{i,IF}$.
\end{definition}

Since then we require that $\Sigma_P\subseteq\bigcup_{i\in I} \Sigma_i$, and we use the notation $A_i^w$ to denote the weakest assumption for component $M_i$ over $\Sigma_{i,IF}$ being set to $\Sigma_{A_i}$ such that $\Sigma_{A_i}\subseteq (\bigcap_{i\in I} \Sigma_i)\cup \Sigma_P (i\in I)$.

\subsection{$L^*$ Learning Algorithm}
Angluin \cite{ang} developed the $L^*$ learning algorithm (abbreviated as $L^*$ algorithm hereafter) for solving the problem of identifying an unknown regular language $U$ over an alphabet $\Sigma$ and of constructing a minimal DFA\footnote{By ``minimal" we mean that the obtained DFA contains the least number of states.} that accepts it. Interested readers are referred to \cite{ang} for more details.

The $L^*$ algorithm interacts with a {\it minimally adequate Teacher}, henceforth deemed as the {\it Teacher}, that answers two types of questions. The first type is referred to as {\it membership queries}, i.e., the Teacher justifies whether or not a word $s\in \Sigma^*$ belongs to $U$; the second type is a {\it conjecture}, in which the $L^*$ algorithm constructs a DFA $M$ and the Teacher justifies $L_m(M)=U$ or not; the Teacher returns a counterexample $c\in L_m(M)\Delta U$ whenever the conjecture is denied.

In order to construct a conjectured DFA $M$, the $L^*$ algorithm incrementally collects information about a finite collection of words in $\Sigma^*$ and records their membership status in an {\it observation table}. The observation table is a three-tuple $(S,E,T)$, where $S$ and $E$ are a set of prefix-closed and suffix-closed words over $\Sigma$, respectively, and $T: (S\cup S\Sigma)E\to\{0,1\}$ is the {\it membership function} that maps words in $s\in (S\cup S\Sigma)$ onto $1$ if they are members of $U$, otherwise it returns $0$. The observation table can be viewed as a 2-dimensional array whose rows are labeled by words $s\in S\cup S\Sigma$ and whose columns are labeled by symbols $e \in E$. The entries in the labeled rows and columns are evaluated by the function $T(se)$. If $s\in S\cup S\Sigma$, then the {\it row function} $row(s)$ denotes the finite function $f$ from $E$ to $\{0,1\}$ defined by $f(e)=T(se)$. The following two properties are essential in the $L^*$ algorithm.

%The observation table constructed by $L^*$ after the $i$-th iteration will be denoted as $T^i$.

\begin{definition}[Closeness and Consistency]\label{cc}
\cite{ang} An observation table is said to be {\it closed} if $(\forall s\in S)(\forall \sigma\in\Sigma)$ $[\exists s'\in S: row(s\sigma)=row(s')]$. It is said to be {\it consistent} if $(\forall s_1, s_2\in S: row(s_1)=row(s_2))(\forall \sigma\in \Sigma)[row(s_1\sigma)=row(s_2\sigma)]$.
%\begin{itemize}
%\item {\it closed} if $(\forall s\in S)(\forall \sigma\in\Sigma)$ $[\exists s'\in S: row(s\sigma)=row(s')]$.
%\item {\it consistent} if $(\forall s_1, s_2\in S: row(s_1)=row(s_2))(\forall \sigma\in \Sigma)[row(s_1\sigma)=row(s_2\sigma)]$.
%\end{itemize}
\end{definition}

If an observation table is not closed, then $s\sigma$ is added to $S$ and $T$ is updated to make it closed where $s\in S$ and $\sigma\in\Sigma$ are the elements for which no $s'\in S$ exists. If an observation table is not consistent, then there exist $s_1, s_2\in S$, $\sigma\in \Sigma$ and $e\in E$ such that $row(s_1)=row(s_2)$ but $T(s_1\sigma e)\ne T(s_2\sigma e)$; to make it consistent, $\sigma e$ is added to $E$ and $T$ is updated.
%\allowdisplaybreaks
%\begin{algorithm}[H]\label{L_algo}
%\caption{$L^*$ \cite{ang}}
%\begin{algorithmic}[l]
%\REQUIRE Alphabet $\Sigma$ and membership quereis $T$
%\ENSURE A DFA $M$ such that $L_m(M)=K$
%\STATE Set $S=\epsilon$ and $E=\epsilon$.
%\STATE Form the initial observation table $T_i(S,E,T)$ where $i=1$
%\WHILE {$T_i(S,E,T)$ is not completed}
%\IF { $T_i$ is not consistent }
%\STATE find $s_1, s_2\in S$, $\sigma\in \Sigma$ and $e\in E$ such that $row(s_1)=row(s_2)$ but $T(s_1\sigma e)\ne T(s_2\sigma e)$;
%\STATE Add $\sigma e$ to $E$;
%\STATE extend $T_i$ to $(S\cup S\Sigma)E$ using membership queries to the oracle.
%\ENDIF
%\IF { $T_i$ is not closed }
%\STATE find $s_1\in S$, $\sigma\in \Sigma$ such that $row(s_1\sigma)$ is different from $row(s)$ for all $s\in S$;
%\STATE Add $s_1\sigma$ to $S$;
%\STATE extend $T_i$ to $(S\cup S\Sigma)E$ using membership queries to the oracle.
%\ENDIF
%\ENDWHILE
%\STATE Let $M_i=M(T_i)$ as the conjectured DFA
%\STATE Ask the Teacher the validity of $M_i$.
%\IF {A counterexample $t\in \Sigma^*$ is provided }
%\STATE Add $t$ and all its prefixes into $S$;
%\STATE extend $T_i$ to $(S\cup S\Sigma)E$ using membership queries to the oracle;
%\ENDIF
%\STATE $i=i+1$ and return to {\bf while} until no more counterexamples.
%\RETURN $M_i$.
%\end{algorithmic}
%\end{algorithm}
%The working procedure of Angluin's $L^*$ is summarized as Algorithm 1.

Once an observation table is both closed and consistent, a candidate DFA $M(S,E,T)=(Q,\Sigma,q_0,\delta,Q_m)$ over the alphabet $\Sigma$ is constructed as follows: $Q=\left\{ {row(s): s\in S} \right\}$, $q_0=row(\epsilon)$, $Q_m=\left\{ {row(s): (s\in S)\land (T(s)=1)}\right\}$, and $\delta(row(s),\sigma)=row(s\sigma)$.
%$$
%\begin{array}{l}
%Q=\left\{ {row(s): s\in S} \right\}, \\
%q_0=row(\epsilon), \\
%Q_m=\left\{ {row(s): (s\in S)\land (T(s)=1)}\right\}, \\
%\delta(row(s),\sigma)=row(s\sigma).
%\end{array}
%$$

The Teacher takes $M$ as a conjecture and if $L_m(M)=U$, the Teacher returns ``True" with the current DFA $M$; otherwise, the Teacher returns ``False" with a counterexample $c\in L_m(M)\Delta U$. The $L^*$ algorithm adds all the words in $\overline{\{c\}}$ to $S$ and iterate the entire procedure to update a new closed and consistent observation table.
The following theorem asserts that the sequence of DFAs constructed by the $L^*$ algorithm to be consistent\footnote{A DFA $M$ is said to be consistent with the function $T$ if for every $s\in S\cup S\Sigma$ and $e \in E$, $\delta(q_0,se)\in Q_m$ if and only if $T(se)=1$.} with $T$ after each counterexample examination is strictly increasing in the number of states.

\begin{theorem}\label{L_iso}
\cite{ang} If $(S,E,T)$ is closed and consistent, and $M(S,E,T)$ is a DFA constructed from $(S,E,T)$. Let $(S',E',T')$ be the updated closed and consistent observation table if a counterexample $t$ is added to $(S,E,T)$. If $M(S,E,T)$ has $n$ states, then the DFA $M(S',E',T')$ constructed from $(S',E',T')$ has at least $n+1$ states.
\end{theorem}
%
%It is worth pointing out that the conjectures made by the $L^*$ strictly increase in size; each conjecture is smaller than the next one, and all incorrect conjectures are smaller than $M$.

The correctness and finite convergence of the $L^*$ algorithm are assured by the following theorem.

\begin{theorem}\label{L_ter}
\cite{ang} Given an unknown regular language $U\subseteq \Sigma^*$, the $L^*$ algorithm eventually terminates and outputs an DFA isomorphic to the minimal one accepting $U$. Moreover, if $n$ is the number of states of the minimal DFA accepting $U$ and $m$ is an upper bound on the length of any counterexample provided by the Teacher, then the total running time of the $L^*$ algorithm is bounded by a polynomial in $m$ and $n$.
\end{theorem}

\section{Problem Formulation and Overall Approach}\label{sec:formulation}
\subsection{Automata-based Agent and Environment Models}
We study the coordination and control of multiple agents moving in a shared and uncertain environment that admits uncertain elements.
%In the context of cooperative multi-agent systems, we use $\Sigma$ to denote the set of action and motion events of the agents. To characterize the behavior of an individual agent in the discrete-event format
For a multi-agent system $G$ that consists of $n$ agents, we model the mission executions of agent $G_i (i\in I)$ as a DFA
\begin{equation}
G_i=(Q_i,\Sigma_i,q_{i,0},\delta_i,Q_{i,m}).
\end{equation}
The local mission set $\Sigma_i$ is partitioned into local controllable missions $\Sigma_{i,c}$ and local uncontrollable missions $\Sigma_{i,uc}$. We are interested in the generated language of $G_i$ and we assume that $Q_{i,m}=Q_i$. The global events, controllable and uncontrollable events are given by $\Sigma=\bigcup_{i\in I} \Sigma_i$, $\Sigma_c=\bigcup_{i\in I} \Sigma_{i,c}$ and $\Sigma_{uc}=\Sigma-\Sigma_c=\bigcap_{i\in I} \Sigma_{i,uc}$, respectively. The collective behavior of the system $G$ is captured by $G=\vert\vert_{i\in I} G_i$. For a given mission $\sigma\in\Sigma$, $I_\sigma=\{i\in I| \sigma\in\Sigma_i\}$. $|I_\sigma|>1$ suggests that all agents $G_i$ with $i\in I_\sigma$ cooperate to accomplish the mission $\sigma$.

All the agents are assumed to move in a shared environment that is partitioned into $N$ regions, and we assume that there are $N_d$ doors placed among these regions. Two adjacent regions in the environment may be connected by one or more doors. The cooperative multi-agent system is provided with the nominal environment {\it a priori}, which is modeled as a 4-tuple
\begin{equation}
\mathcal{E}_0=(V,\longrightarrow_{\mathcal{E}_0},D,F_D),
\end{equation}
where $V=\{v_1,v_2,\ldots,v_N\}$ is the set of labels for the partitioned regions, $\longrightarrow_{\mathcal{E}_0}\subseteq V\times V$ is the adjacency relation among the regions, $D=\{d_1,d_2,\ldots,d_{N_d}\}$ is the set of labels for the doors, $F_D: V\times V \to 2^D$ is a partial mapping such that for all $(v,v')\in\longrightarrow_{\mathcal{E}_0}$, $F_D(v,v')\subseteq D$ represents all the possible doors that connect the adjacent regions $v$ and $v'$, whereas $F_D(v,v')=\emptyset$ implies that $v$ and $v'$ are not connected by any doors.

The motion capabilities of agent $G_i (i\in I)$ within the nominal environment $\mathcal{E}_0$ are abstracted as a trim DFA:
\begin{equation}
G^m_i=(V,D,v_{i,0},\delta^m_i,V_m),
\end{equation}
where $V$ and $D$ are defined in $\mathcal{E}_0$, $v_{i,0}\in V$ is the initial region of agent $G_i$, $V_m=V$ and the motion transition $\delta^m_i$ is defined as follows:
\begin{equation*}
\forall v\in V, d\in D: \delta^m_i(v,d)=v' \Leftrightarrow d\in F_D(v,v'),
\end{equation*}
i.e., an agent can move to an adjacent region whenever there is a door connecting the two regions.

The uncertainty of the multi-agent system arises in two aspects. Firstly, the practical environment $\mathcal{E}$ may differ from the nominal one $\mathcal{E}_0$, in the sense that some transitions in $F_D$ may become unaccessible. Secondly, the mission capabilities of $G_i$'s need not be given {\it a priori}.

\subsection{Coordination of Cooperative Multi-agent Systems}
The integration of mission and motions of the agent $G_i$ is captured by the following local labeling mapping
\begin{equation}
\pi_i: \Sigma_i \to V,
\end{equation}
that specifies a region in $V$ in which a mission in $\Sigma_i$ shall be performed.
%The definition of $\pi_i$ implies that an agent can execute various missions in one region, whereas a fixed mission cannot be executed in different regions.
%; in particular, the symbol $\epsilon$ represents the case that the agent executes no mission in a certain region.

The coordination of the agents accounts for both mission and motion plans. Formally, an integrated mission-motion plan is defined as follows.

\begin{definition}[Integrated Local Plans]
A local integrated mission-motion plan for agent $G_i (i\in I)$ is a prefix-closed language $LP_i\subseteq (V\cup \Sigma_i)^*$ such that
\begin{enumerate}
\item $LP_i(0) = v_{i,0}$;
\item If $LP_i(k)\in \Sigma_i$ and $LP_i(k-1)\in V$, then $LP_i(k-1) \in \pi_i(LP_i(k))$; otherwise, $\pi_i(LP_i(k-1)) = \pi_i(LP_i(k))$
\item $P_V(LP_i)\subseteq Run\left[L(G^m_i)\right]$, where $P_V$ denotes the natural projection from $V\cup\Sigma_i$ to $V$ and $Run\left[L(G^m_i)\right]=\bigcup_{s\in L(G^m_i)} Run(s)$.
\end{enumerate}
\end{definition}

We use $L_i^{mo}=P_V(LP_i)$ to denote the {\it motion plan} component of an integrated local plan $LP_i$ for agent $G_i$, whereas $L_i^{mi}=P_\Sigma(LP_i)$ stands for the {\it mission plan} component of $LP_i$, where $P_\Sigma$ is the natural projection from $V\cup\Sigma_i$ to $\Sigma_i$.

Our main objective is concerned with the synthesis of local mission plans and corresponding motion plans for each agent such that a global mission can be accomplished.
\begin{problem}\label{dcccp}
Given a cooperative multi-agent system $G$ that consists of $n$ agents with local mission DFA $G_i$ and motion DFA $G_i^m (i\in I)$ within the nominal environment $\mathcal{E}_0$, and a prefix-closed global mission $L\subseteq \Sigma^*$ whose independence relation (cf. Remark 2) is transitive, solve the following problems:
\begin{enumerate}
\item {\it Mission assignment:} systematically find a feasible local mission specification $L^{mi}_i$ for each agent $G_i$;
\item {\it Mission planning:} synthesize a series of local supervisors $S_i (i\in I)$ such that $L(S_i||G_i)=L_i^{mi}$;
\item {\it Motion planning:} synthesize local motion plans $L_i^{mo} (i\in I)$ such that the combination of $L^{mi}_i$ and $L^{mo}_i$ forms an integrated local plan $LP_i$;
\item {\it Mission-motion integration:} $\vert\vert_{i\in I} LP_i \models L$.
\end{enumerate}
\end{problem}

\subsection{Overview of the Approach}
Without loss of generality, we assume that $L\subseteq L(G)$. We propose a two-layer formal synthesis framework to solve Problem~\ref{dcccp} via ``divide-and-conquer". The synthesis procedure of the framework is depicted in Fig.~\ref{caf}. In addition, regular language learning algorithms are utilized in the proposed framework to guarantee the performance of the cooperative multi-agent system in the presence of system and/or environment uncertainties. The fact that $\Sigma\cap D=\emptyset$ enables us to solve the mission and motion planning sub-problems in a modular fashion, which can be summarized as iterative execution of the following procedure.

\begin{figure}[t]
\begin{center}
    \centerline{\includegraphics[width=0.40\textwidth]{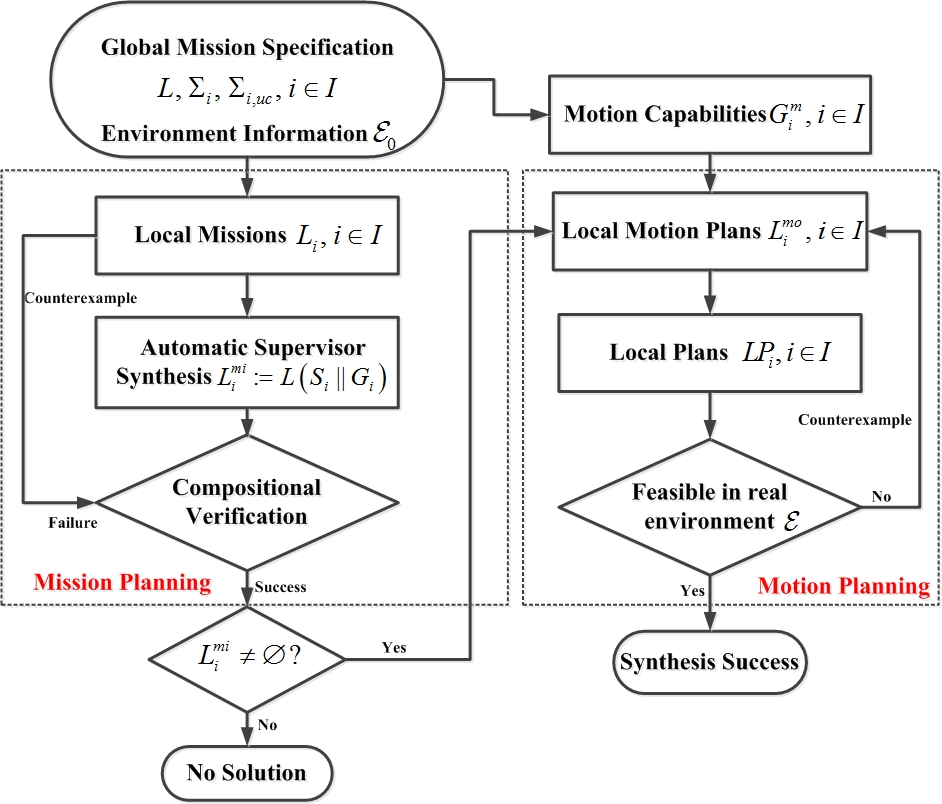}}
    \caption{Learning-based coordination and control framework.}
    \label{caf}
  \end{center}
  \vspace{-12.5mm}
\end{figure}

{\it 1) Mission decomposition:} Obtain a prefix-closed and locally feasible mission specification $L_i$ for agent $G_i (i\in I)$.

{\it 2) Local supervisor synthesis:} Given the mission specification $L_i$, we use a modified $L^*$ algorithm to automatically synthesize a local supervisor $S_i$ such that $L(S_i\vert\vert G_i)=\sup C_i(L_i)$, where $\sup C_i$ stands for the supremal controllable sublanguage with respect to $G_i$ and $\Sigma_{i,uc}$.

{\it 3) Compositional verification:} We employ a compositional verification strategy to justify whether or not all the controlled agents can fulfill the global mission. Specifically, we apply the assume-guarantee paradigm by setting $M_i:=S_i\vert\vert G_i (i\in I)$ as component modules and $L$ as the property to be verified. The verifier returns a counterexample $t\subseteq \Sigma^*$ that indicates the violation of $L$ when the compositional verification fails. Such counterexample $t$ is applied for the re-synthesis of the local supervisors until no more counterexamples are detected. Afterwards, local mission plans $L^{mi}_i (i\in I)$ are obtained.

{\it 4) Automatic motion planning:} Based on the synthesized mission plans $L^{mi}_i$, corresponding motion plans $L^{mo}_i$ are generated by incorporating $G^m_i (i\in I)$ such that the combination of $L^{mi}_i$ and $L^{mo}_i$ forms a valid integrated local plan $LP_i$.

{\it 5) Counterexample-guided motion replanning:} Agent $G_i$ implements $LP_i$ in the practical environment $\mathcal{E}_0$; once there exists a mismatch between $\mathcal{E}$ and $\mathcal{E}_0$ that makes the current motion plan $L^{mo}_i$ is infeasible, a counterexample is automatically generated to refine $L^{mo}_i$ such that a new integrated plan $LP_i$ can be executed by $G_i$.

%, suggesting that all the $M_i, i\in I$ share a same illegal word that violates property $K$

The following sections illustrate the mission and motion planning layers respectively.

\section{Learning-based Compositional Synthesis of Mission Plans}\label{sec:synthesis}
This section concerns with the (offline) mission planning of the proposed framework shown in in Fig.~\ref{caf}. Synthesis of mission plans $L^{mi}_i (i\in I)$ consists of two stages: first, we exploit the local mission capabilities $\Sigma_i$ to develop feasible local missions and mission supervisors; next, we compositionally verify the joint performance of the coordination of the mission plans by using the assume-guarantee scheme \cite{cob2}.
\subsection{Initiation of the Mission Plans}
We originate the generation of local missions $L_i (i\in I)$ by setting $L_i=P_i(L)$. The local feasibility of $L_i$ is assured by the following theorem.

\begin{theorem}\label{fea}
Given a cooperative multi-agent system $G$ that consists of $n$ agents $G_i (i\in I)$ and a non-empty and prefix-closed specification $L\subseteq L(G)$, languages $L_i:=P_i(L) (i\in I)$ satisfy $L_i\subseteq L(G_i)$ and $L_i=\overline{L_i}$.
\end{theorem}

\begin{proof}
From Definition~\ref{product}, the prefix-closeness of $L$ implies that $P_i(L)$ is prefix-closed. Thus we obtian $L_i=\overline{L}_i$. Next, according to Definitions~\ref{parallel} and~\ref{product},
\begin{equation}\nonumber
\begin{split}
L_i = P_i(L)\subseteq P_i(L(G))=P_i\left[\bigcap_{i\in I}P_i^{-1}[P_i(L(G))]\right] \\
\subseteq P_i(P_i^{-1}[P_i(L(G))])=P_i(L(G))=L(G_i),
\end{split}
\end{equation}
which completes the proof.
\end{proof}

\subsection{Automatic Synthesis of Local Mission Supervisors}
We aim to synthesize a maximally permissive \cite{cas} local mission supervisor $S_i$ that drives $G_i$ to fulfill $L_i (i\in I)$. According to Definition~\ref{controllability}, applying directly the $L^*$ algorithm to synthesize an appropriate supervisor for agent $G_i (i\in I)$ is infeasible when full knowledge of $G_i$ is not given {\it a priori}. To overcome this computational burden, we modify the algorithm and propose a learning-based algorithm, namely the $L^*_{LS}$\footnote{The subscript ``LS" stands for ``local synthesis".} algorithm, for agent $G_i$ to learn $\sup C_i(L_i)$ with the local controllability information $\Sigma_i=\Sigma_{i,c}\dot\cup\Sigma_{i,uc}$ rather than $G_i$.

The local mission supervisors are constructed on the basis of illegal words. A mission $s\in L(G_i)$ is said to be {\it illegal} if $s\not\in L_i$. A local mission $st\in L(G_i)$ is said to be {\it uncontrollably illegal} if it becomes illegal due to an uncontrollable suffix, i.e., $s\in L_i$ , $t\in \Sigma_{i,uc}^*$ and $st\not\in L_i$. Let $C_i$ be the collection of uncontrollably illegal behaviors of agent $G_i (i\in I)$. Define $D_{ui}(\cdot)$ as an operator to collect words that are formed by discarding the uncontrollable suffixes of mission behaviors in $C_i$, i.e.,
$
D_{ui}(C_i)=\left \{s\in L(G_i)|(\exists t\in \Sigma_{uc}^*)[st\in C_i]\right \}.
$

From Section~\ref{sec:preliminaries}.D, the $L^*$ algorithm sequentially constructs a series of observation tables to generate a correct DFA that recognizes an unknown regular language. This idea is inherited by the $L^*_{LS}$ algorithm and the observation table is formed with help of $C_i$. Specifically, let $C_i^j$ denote the set of uncontrollably illegal behavior after the $j$-th update of the observation table in the $L^*$ algorithm; thus $C_i^{j+1}=\left\{s_j\right\}\cup C_i^j$ if a new uncontrollably illegal behavior $s_j$ is added to $C_i^j$.

The main difference between the $L^*$ and $L_{LS}^*$ algorithms lies on the membership queries and generation of counterexamples. For $t\in\Sigma_i^*$, we let $T_i^j(t) (j\in \mathbb{N})$ be the membership function for $G_i (i\in I)$. Initially, the Teacher justifies the membership of $t$ with respect to $L_i$,
\begin{equation}\label{ls1}
T_i^1(t) =
\begin{cases}
0,  & \mbox{if }t\notin L_i, \\
1, & \mbox{otherwise.}
\end{cases}
\end{equation}
and for $j > 1$, rather than answering the membership question with respect to the unknown $\sup C_i(L_i)$, the Teacher provides answers for the following queries:
\begin{equation}\label{ls2}
T_i^j(t) =
\begin{cases}
0,  & \mbox{if }T_i^{j-1}(t)=0 \mbox { or } t\in D_{ui}(C_i^j)\Sigma_i^*, \\
1, & \mbox{otherwise,}
\end{cases}
\end{equation}

Furthermore, due to the lack of prior knowledge of $G_i$, we define the following sequence $\{K_j\} (j\in \mathbb{N})$ to facilitate the generation of  appropriate counterexamples.
\begin{equation}\label{ls3}
\begin{split}
K_1&:=L_i,\\
K_{j}&:=K_{j-1}-D_{ui}(C_i^j)\Sigma_i^*.
\end{split}
\end{equation}

With slightly abusing the notations, we also use $K_j$ to denote the DFA that recognizes the language $K_j (j\in\mathbb{N})$. Let $M_j=M(S^j,E^j,T^j_i)$ be the DFA that is consistent with membership function $T^j_i$, a word $t\in\Sigma_i$ is a counterexample of the $L^*_{LS}$ algorithm with respect to agent $G_i (i\in I)$ if
\begin{equation}\label{ls5}
t\in L(M_j)\Delta K_j.
\end{equation}

The $L^*_{LS}$ algorithm can be viewed as an execution of the $L^*$ algorithm by using membership queries in (\ref{ls1})(\ref{ls2}) while generating counterexamples by examining (\ref{ls5}). Details of the $L^*_{LS}$ algorithm are presented in Algorithm 1.
\subsection{Correctness of the $L^*_{LS}$ Algorithm}
\begin{algorithm}[t]
%\captionsetup[algorithm]{name=$L^*_{LS}$ Algorithm}
\caption{$L^*_{LS}$ Algorithm}
\begin{algorithmic}[1]
\REQUIRE Local mission specification $L_i$, local event sets $\Sigma_{i,c}$ and $\Sigma_{i,uc}$
\ENSURE Local supervisor $S_i$ such that $L(S_i\vert\vert G_i)=L(S_i)=\sup C_i(K_i)$
\STATE $S=\{\epsilon\}$, $E=\{\epsilon\}$, $j= 1$
\STATE Construct $T^1_i(S,E,T)$ by answering membership queries (\ref{ls1})
\REPEAT
\WHILE {$T^j_i(S,E,T)$ is not closed or consistent}
\IF { $T^j_i$ is not closed }
\STATE find $s\in S$, $\sigma\in \Sigma_i$ such that $\forall s'\in S: row(s'\sigma)\ne row(s)$
\STATE $S = S\cup\{s\sigma\}$
\STATE extend $T^j_i$ to $(S\cup S\Sigma_i)E$ using membership queries (\ref{ls1}) or (\ref{ls2})
\ENDIF
\IF { $T^j_i$ is not consistent }
\STATE find $s_1, s_2\in S$, $\sigma\in \Sigma_i$ and $e\in E$ such that $row(s_1)=row(s_2)$ but $T(s_1\sigma e)\ne T(s_2\sigma e)$
\STATE $E = E\cup\{\sigma e\}$
\STATE extend $T^j_i$ to $(S\cup S\Sigma_i)E$ using membership queries  (\ref{ls1}) or (\ref{ls2})
\ENDIF
\ENDWHILE
\STATE $M_j=M(S^j,E^j,T^j_i)$
%\STATE Make the conjecture $L_m(M_j)=\sup C_i(L_i)$.
\IF { the Teacher presents a counterexample $t\in \Sigma_i^*$ according to (\ref{ls5})}
\STATE $S = S \cup \overline{t}$
\STATE $j=j+1$
\STATE extend $T^j_i$ to $(S\cup S\Sigma_i)E$ using membership queries (\ref{ls1}) or (\ref{ls2})
\ENDIF
\UNTIL {the Teacher generates no more counterexamples}.
\RETURN $S_i=M_j$
\end{algorithmic}
\end{algorithm}

The following lemma helps us investigate the correctness and finite convergence properties of the $L_{LS}^*$ algorithm.

\begin{lemma}\label{supC(K) iteration}
(Paraphrased from \cite{kum}) If $L_i$ and $L(G_i)$ are regular languages, then $\sup C_i(L_i)$ can be computed iteratively as follows.
\begin{equation}\label{ls6}
\begin{split}
K_1&:=L_i,\\
K_{j+1}&:=K_j-[(L(G_i)-K_j)/\Sigma_{i,uc}]\Sigma_i^*.
\end{split}
\end{equation}
If there exists $N\in \mathbb{N}$ such that $K_{N+1}=K_N$, then $\sup C_i(L_i)=K_N$. Moreover, if $L_i$ is prefix-closed, then $\sup C_i(L_i)$ is also prefix-closed and  can be computed directly as
\begin{equation}\label{ls8}
\sup C_i(L_i)=L_i-[(L(G_i)-L_i)/\Sigma^*_{i,uc}]\Sigma_i^*.
\end{equation}
\end{lemma}

We now proceed to the correctness of the $L^*_{LS}$ algorithm. In fact, comparing (\ref{ls6}) with the dynamical membership queries stated in (\ref{ls1})(\ref{ls2}) and the counterexample-related iteration (\ref{ls3}), we observe that initially, for $i\in I$, (\ref{ls6}) is identical to (\ref{ls3}), and $T^1_i(t)$ is consistent with $K_1=L_i$. Furthermore, when the iteration step $j\ge 2$, counterexamples are generated and based on which $C^j_i$ is formed. When no more counterexamples are generated, the set $C^j_i$ eventually becomes identical to $L(G_i)-L_i$, which includes all the local illegal and uncontrollably illegal behaviors and it turns out that the language consistent with $T^j_i$ is
\begin{equation}\label{ls9}
L_i-D_{ui}(C^j_i)\Sigma_i^*=L_i-[D_{ui}(L(G_i)-K_i)]\Sigma_i^*,
\end{equation}
Note that the right side of (\ref{ls9}) coincides with the results obtained in (\ref{ls8}), which implies that the output of the $L^*_{LS}$ algorithm is indeed $\sup C_i(L_i)$, provided that it is convergent.

Next we will prove that the $L^*_{LS}$ algorithm synthesizes $\sup C_i(L_i)$ within a finite number of iterations. From Algorithm 1, there are two kinds of ``counterexamples" that are used by the $L^*_{LS}$ algorithm: counterexamples used to make an observation table $T^j_i$ closed and consistent (lines 4-15), and counterexamples generated by the Teacher (\ref{ls5}) to construct a new observation table $T^{j+1}_i$ (lines 17-20). We denote $D_i$ as the collection of the ``true" counterexamples that are generated according to (\ref{ls5}), while $D'_i$ is used to denote the set of counterexamples that help make an observation table closed and consistent.
%make a closed and consistent $T_i$ after a new illegal string $s$ is detected and added to update $S$ of the current observation table $(S,E,T)$.

Consider an observation table $T=(S^j_i,E^j_i,T^k_i)$ with the corresponding DFA $M(S^j_i,E^j_i,T^k_i) = M$, where $S^j_i$ and $E^j_i$ are the pre-defined $S$ and $E$ after the $j$-th counterexample from $D_i$ has been added to the table, and $T^k_i$ denotes the current membership function after $k$-th counterexample in $D'_i$ has been added to update it. Let $n_{k,j}$ denote the number of states of DFA $M$, $n_k$ denote the number of states of the minimal DFA that is consistent with $T^k_i$, and $n^\uparrow$ denote the number of states of the minimal DFA that recognizes $\sup C_i(L_i)$. According to membership queries (\ref{ls1})(\ref{ls2}) and Lemma~\ref{supC(K) iteration}, the sequence of DFAs that are consistent with $T_i^k$ finally recognizes $\sup C_i(L_i)$, which implies that $n_k \to n^\uparrow$ (with respect to $k$) and the convergence is achieved within a finite number of iterations.

In addition, from Theorem~\ref{L_iso}, $n_{k,j}\le n_k$ for all $j\in \mathbb{N}$. Moreover, it is clear that $\left\{n_{k,j}\right\}$ is a monotonically increasing sequence (with respect to $j$). Therefore, $n_{k,j}\to n_k$ for a fixed $k$ (with respect to $j$) is enforced by its monotonicity and upper-boundedness. The fact that $n_k\to n^\uparrow$ thus yields that $n_{k,j}\to n^\uparrow$.
%$\left\{n_k\right\}$ and $\left\{n_{kj}\right\}$ are both monotonically increasing sequences (with respect to $k$ and $j$, respectively.) and $n_{kj}\le n_k$ for any fixed $k$ and all $j\in \mathbb{N}$. From Lemma \ref{supC(K) iteration}, $n_k \to n^\uparrow$ (with respect to $k$) after a finite step of iterations, and $n_{kj}\to n_k$ eventually according to Theorem \ref{L_ter}. Hence we can conclude that the iteration of the $L^*_{LS}$ using membership queries (2) and (3) shall always converge after the Teacher generates a finite number of counterexamples.

The above discussion in fact proves the correctness and finite convergence properties of the $L^*_{LS}$ algorithm, as summarized in the following theorem.

\begin{theorem}\label{LC}
For a non-empty, prefix-closed and locally feasible mission specification $L_i \subseteq L(G_i)$ for agent $G_i (i\in I)$, the $L^*_{LS}$ algorithm with dynamical membership queries (\ref{ls1}) and (\ref{ls2}), and counterexample generated from (\ref{ls5}), converges to a local mission supervisor $S_i$ such that $L(S_i\vert\vert G_i)=\sup C_i(L_i)$. Moreover, the $L_{LS}^*$ algorithm can always construct the correct $S_i$ after a finite number of counterexample tests.
\end{theorem}

\begin{remark}
It follows from Theorem 3 that the $L^*_{LS}$ algorithm also admits a polynomial-time computational complexity. Compared to the conventional supervisor synthesis algorithms proposed in \cite{cas} and \cite{kum}, the presented $L_{LS}^*$ algorithm requires no prior knowledge of the plant.
\end{remark}
%
%\begin{remark}
%One shall find that the synthesis of local supervisors with the $L_{LS}^*$ requires only local information about each individual agent's capabilities $\Sigma_i$.
%\end{remark}

\subsection{Automated Compositional Verification via Learning}
Let $L^{mi}_i:=\sup C_i(L_i)$ be the (initial) local mission plan of agent $G_i (i\in I)$ in the cooperative multi-agent system. In case that the global mission $L$ is prefix-closed and separable with respect to $\{\Sigma_i\}_{i\in I}$, the mission assignment and planning for each agent can be solved in a modular manner by synthesizing local supervisors independently to fulfill the initial local specifications $L_i=P_i(L) (i\in I)$ using the $L^*_{LS}$ algorithm, since according to \cite{will}, separability of $L$ implies $L=\vert\vert_{i\in I} P_i(L)$. Nonetheless, whenever $L$ fails to be separable, it is non-trivial to determine whether or not the concurrent execution of $L^{mi}_i (i\in I)$ satisfies the global mission $L$. To automate the assume-guarantee strategy, we propose another modification of the $L^*$ algorithm, namely the $L^*_{CV}$\footnote{The subscript ``CV" stands for ``compositional verification".} algorithm, to learn appropriate and weakest local assumptions for each supervised agent

\begin{equation}\label{cv1}
M_i:=S_i||G_i (i\in I).
\end{equation}

Let $M=\vert\vert_{i\in I}M_i$ denote the mission performance of the controlled cooperative multi-agent system, $A_i$ denote an assumption about $M_i$'s environment. With slightly abusing the notations, we use $L$ and $coL$ to denote the DFA that recognizes $L$ and its complement language, respectively. The $L^*_{CV}$ algorithm checks $M\models L$ via the assume-guarantee proof rules SYM-N. A two-layered implementation of the $L^*_{CV}$ algorithm is depicted in Fig.~\ref{ag}.

As shown in Fig.~\ref{ag}, in the first layer, the $L^*_{CV}$ algorithm deploys $n$ {\it local} Teachers, each of which corresponds to an controlled agent $M_i (i\in I)$, to learn the appropriate and weakest assumption $A_i$ for $M_i$. At the $j$-th iteration step, the $L_{CV}^*$ algorithm learns an appropriate assumption $A_i^j$ for the controlled agent $M_i$, by querying the agent and by reviewing the results of the previous iteration. As the $L_{CV}^*$ algorithm executes the learning procedure iteratively, it incrementally (in the sense of the number of states) constructs a sequence of assumption DFAs $\{A_i^j\}_{j\in \mathbb{N}}$ for $M_i$, which converges to the DFA that recognizes the weakest assumption $A_i^w$ whose event set is constrained to be $\Sigma_{A_i}$ (cf. Section~\ref{sec:preliminaries}.C).

The implementation of the $L_{CV}^*$ algorithm relies on the following lemma.
\begin{lemma}\label{lcv}
(Adapted from \cite{cob2}) Given $t\in \Sigma^*$, $i\in I$ and a regular property $L$, $t\in L(A_i^w)$ if and only if $\langle\mathcal {DFA}(t) \rangle M_i \langle L \rangle$ holds, where $\mathcal {DFA}(t)$ is a trim DFA such that $L(\mathcal {DFA}(t))=L_m(\mathcal {DFA}(t))=\overline{t}$.
\end{lemma}
%
%\begin{proof}
%Let $M_{-i}=\vert\vert_{j\in I, j\ne i}M_j$. By Theorem \ref{assume-guarantee}, $\langle\mathcal {DFA}(t) \rangle M_i \langle L \rangle$ is true if and only if $q_e$ is unreachable in $\mathcal {DFA}(t)\vert\vert M_i \vert\vert coL $, which is equivalent to checking $\langle true \rangle \mathcal {DFA}(t)\vert\vert M_i \langle L \rangle$. From Definition \ref{weakest assumption}, this is the same as checking $\langle true \rangle \mathcal {DFA}(t) \langle A_i^w \rangle$ (by substituting $M_{-i}$ with $\mathcal {DFA}(t)$), which is clearly equivalent to checking $t\in L(A_i^w)$, completing the proof.
%\end{proof}

\begin{figure}[t]
\begin{center}
    \centerline{\includegraphics[width=0.36\textwidth]{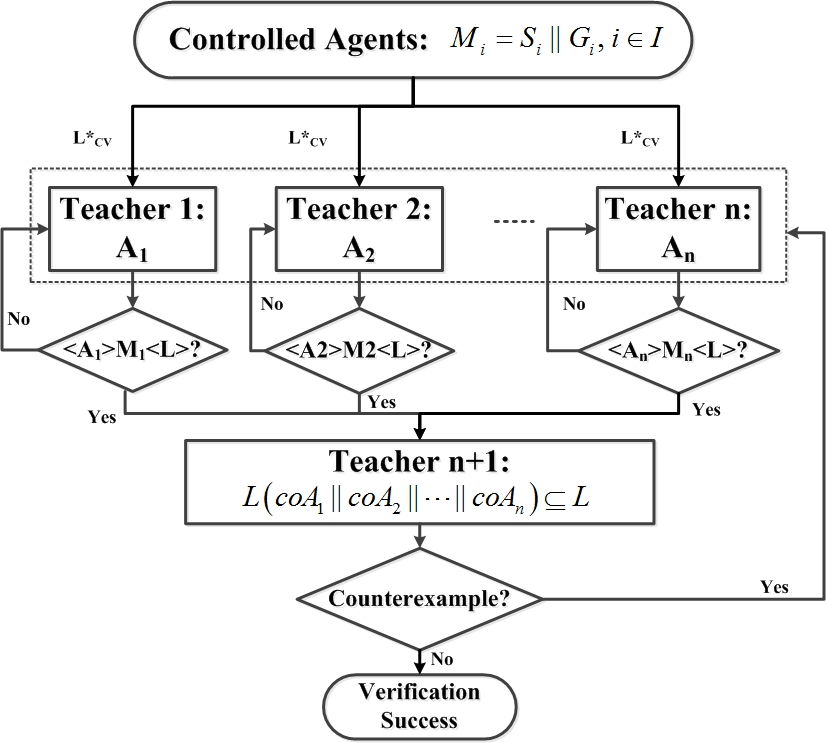}}
    \caption{Compositional verification via assumption learning.}
    \label{ag}
  \end{center}
  \vspace{-10mm}
\end{figure}

Suggested by Lemma~\ref{lcv}, each local Teacher of the $L_{CV}^*$ algorithm answers the following membership queries to construct local assumption DFAs for each controlled agent $M_i$: for $i\in I$ and $t\in \Sigma^*$,
\begin{equation}\label{L_cv}
T_i(t) =
\begin{cases}
1,  & \mbox{if } \langle\mathcal {DFA}(t) \rangle M_i \langle L \rangle \mbox{ is true,} \\
0, & \mbox{otherwise.}
\end{cases}
\end{equation}

In addition to the membership queries (\ref{L_cv}), the Teacher of the $L^*_{CV}$ algorithm justifies the conjecture $\langle A_i^j \rangle M_i \langle L \rangle$ rather than $L(A_i^j) = L(A_i^w)$. Once the Teacher denies the conjecture, a counterexample $t\in \Sigma^*$ is proposed by the Teacher, and the $L^*_{CV}$ algorithm adds $P_{A_i}(t)$ and all its prefixes back to the local iteration loop to update the observation tables, where $P_{A_i}$ is the natural projection from $\Sigma^*$ to $\Sigma_{A_i}^*$.
%
%\begin{theorem}\label{AG-C rule}
%\cite{cob2} The circular and symmetric assume-guarantee rule (SAGR) is sound and complete.
%\end{theorem}
%
%Theorem \ref{AG-C rule} ensures the completeness and soundness of the SAGR. Intuitively, the soundness of the SAGR will imply the correctness of the compositional verification part of the framework shown in Fig.~\ref{caf}, while the completeness guarantees finite convergence of the verification procedure.
%
%Based on Theorems \ref{lcv} and \ref{AG-C rule}, we can conclude that the assume-guarantee reasoning compositional verification in Fig.~\ref{ag} empolying the $L_{CV}^*$ queries (8) is correct and convergent, which is formally summarized as the following theorem.

Once the local learning loop of the $L^*_{CV}$ algorithm terminates, we collect a family of local assumptions $A_i:=A_i^w (i\in I)$. In addition to the $n$ local Teachers, the $L^*_{CV}$ algorithm deploys Teacher $n+1$ in the second layer in the framework shown in Fig.~\ref{ag} to justify whether $L(coA_1\vert\vert\cdots\vert\vert coA_n)\subseteq L$ or not. If Teacher $n+1$ returns ``True", the synthesis framework terminates with the conclusion that $\vert\vert_{i\in I} M_i\models L$. Otherwise, Teacher $n+1$ returns ``False" with a counterexample $t\in \Sigma^*$. The $L^*_{CV}$ algorithm then determines whether or not the global mission $L$ is indeed violated by the collective behavior of $M$, which is performed by simulating $t$ on each composed DFA $M_i\vert\vert coL (i\in I)$ and by checking whether or not $t$ can be accepted. If $t$ is not a violating trace for at least one agent $M_i$, we treat $t$ in the same way in the first layer and use $t$ to re-construct the local assumption $A_i$ for controlled agent $M_i$. Otherwise, $t$ turns out to be a common violating word of all agents and the controlled cooperative multi-agent system $\vert\vert_{i\in I} M_i$ indeed violates $L$. As shown in Fig.~\ref{caf}, undesired mission behaviors emerge in the joint execution of synthesized mission plans $L^{mi}_i$ for agent $G_i (i\in I)$; therefore, the re-synthesis of local mission plans is triggered (cf. Section~\ref{sec:synthesis}.E). The correctness and termination properties of the $L^*_{CV}$ algorithm are summarized in the following theorem.

\begin{theorem}\label{cv correctness}
\cite{cob2} For the global mission $L$ and controlled agents $M_1, M_2, \cdots, M_n$, the $L_{CV}^*$ algorithm implemented by the framework shown in Fig.~\ref{ag} with the proof rules SYM-N terminates within finite number of iterations and correctly returns whether or not $\vert\vert_{i\in I}M_i\models L$.
\end{theorem}

\subsection{Counterexample-guided Re-synthesis of Mission Plans}
Our proposed mission planning scheme explores the re-synthesis of the local mission plans after a mission-violating counterexample $t$ is provided by the compositional verification procedure. The re-synthesis is accomplished by the following steps: first,
\begin{equation}\label{co1}
L^{temp}_i = L^{mi}_i-P_i(t);
\end{equation}
next, to apply the $L^*_{LS}$ algorithm for the re-synthesis of the local mission supervisor, it is desired to provide agent $G_i (i\in I)$ with a prefix-closed mission specification; this requirement is fulfilled by setting
\begin{equation}\label{co2}
L_i = L^{temp}_i - coL^{temp}_i\Sigma_i^*,
\end{equation}
where $coL^{temp}_i=\Sigma_i^*-L^{temp}_i$ is the complement language of $L^{temp}_i$. The obtained $L_i$ is used as the new initial local mission for agent $G_i$ to update the new local mission supervisor $S_i (i\in I)$ with the $L^*_{LS}$ algorithm.

Since it has been shown in the previous subsections that both the $L^*_{LS}$ and $L_{CV}^*$ algorithms possess finite convergence, we mainly focus on evaluating the performance of the proposed mission planning schemes in this subsection. We introduce the notion of {\it separate controllability} to characterize the solutions of the multi-agent mission planning.

\begin{definition}\label{separate controllability}
Given the cooperative multi-agent system that consists of $n$ agents $G_i$ with $\Sigma_{i,c}\subseteq \Sigma_i$ being the local controllable events, $i\in I$, a prefix-closed language $L\subseteq \Sigma^*$ is said to be separately controllable with respect to $\Sigma_i, \Sigma_{i,uc}$ and $G_i$ if
\begin{enumerate}
\item $L=\vert\vert_{i\in I}P_i(L)$;
\item $P_i(L)$ is controllable with respect to $G_i$ and $\Sigma_{i,uc}$.
\end{enumerate}
\end{definition}

We assert that the separate controllability is a sufficient and necessary condition for the existence of solutions of the cooperative multi-agent mission planning.

\begin{theorem}\label{sep con}
Given a cooperative multi-agent system $G$ that consists of $n$ agents $G_i (i\in I)$ with local controllable events $\Sigma_{i,c}$ and local uncontrollable events $\Sigma_{i,uc}$, and a non-empty and prefix-closed global mission specification $L\subseteq L(G)$, there exists a series of local mission supervisors $S_i (i\in I)$ for each agent $G_i$ such that $L(S_i\vert\vert G_i)=P_i(L)$ and $\vert\vert_{i\in I} L(S_i\vert\vert G_i)=L$ if and only if $L$ is separately controllable with respect to $\Sigma_i$, $\Sigma_{i,uc}$ and $G_i$.
\end{theorem}

\begin{proof}
($\Leftarrow$) When $L\subseteq L(G)$ is separately controllable with respect to $\Sigma_i, \Sigma_{i,uc}$ and $G_i$, Theorem~\ref{fea} yields that $P_i(L)\subseteq L(G_i)$ is prefix-closed; whereas Definition~\ref{separate controllability} suggests that $P_i(L)$ be controllable with respect to $G_i$ and $\Sigma_{i,uc}$, which implies that for each $i\in I$, there exists a local supervisor $S_i$ such that $L(S_i\vert\vert G_i)=P_i(L)$ \cite{kum}. Therefore, $\vert\vert_{i\in I} L(S_i\vert\vert G_i)=\vert\vert_{i\in I}P_i(L)=L$ is guaranteed according to the separability requirement of Definition~\ref{separate controllability}.

($\Rightarrow$) Conversely, when local supervisors $S_i (i\in I)$ exist such that $L(S_i\vert\vert G_i)=P_i(L)$ and $\vert\vert_{i\in I} L(S_i\vert\vert G_i)=L$, $L=\vert\vert_{i\in I} P_i(L)$ trivially holds. Furthermore, since $L(S_i\vert\vert G_i)=P_i(L)$ is a controlled behavior for agent $G_i$, it is clearly controllable with respect to $G_i$ and $\Sigma_{i,uc}$.
\end{proof}

The following theorem states that in general, the separate controllability is a strictly stronger notion than the combination of global controllability and separability.

\begin{theorem}
For the cooperative multi-agent system $G$ that consists of $G_i (i\in I)$, if a non-empty prefix-closed specification language $L\subseteq L(G)$ is separately controllable with respect to $\Sigma_i, \Sigma_{i,uc}$ and $G_i$, then $L$ is separable with respect to $\{\Sigma_i\}_{i\in I}$ and is controllable with respect to $G$ and $\Sigma_{uc}$.
\end{theorem}

\begin{proof}
First, by Definition~\ref{separate controllability}, the separate controllability of $L$ implies that $L=\vert\vert_{i\in I} P_i(L)$, which assures the separability of $L$. Next we show that $L$ is globally controllable. In fact, $P_i(L)$ is prefix-closed and controllable with respect to $G_i$ and $\Sigma_{i,uc}$, i.e.
$$P_i(L)\Sigma_{i,uc} \cap L(G_i) \subseteq P_i(L),$$
it follows from the monotonicity of $P_i^{-1}$ that
\begin{equation}\nonumber
P_i^{-1}\left[P_i(L)\Sigma_{i,uc} \cap L(G_i)\right]\subseteq P_i^{-1}(P_i(L)).
%& \subseteq P_i^{-1}(P_i(L)\Sigma_{i,uc}).
\end{equation}
Since
$$P_i^{-1}\left[P_i(L)\Sigma_{i,uc} \cap L(G_i)\right]=P_i^{-1}\left[P_i(L)\Sigma_{i,uc}\right]\cap P_i^{-1}(L(G_i)),$$
we have
\begin{equation}\nonumber
\begin{split}
L\Sigma_{i,uc} \cap L(G) &\subseteq P_i^{-1}(P_i(L))\Sigma_{i,uc} \cap P_i^{-1}(P_i(L(G))) \\
&\subseteq P_i^{-1}\left[P_i(L)\Sigma_{i,uc}\right] \cap P_i^{-1}(L(G_i))\\
&=P_i^{-1}\left[P_i(L)\Sigma_{i,uc} \cap L(G_i)\right]\subseteq P_i^{-1}(P_i(L)).
\end{split}
\end{equation}

Note that the last inclusion always holds for any $i\in I$. Therefore
\begin{equation}\nonumber
\bigcap_{i\in I}\left[L\Sigma_{i,uc} \cap L(G)\right] \subseteq \bigcap_{i\in I}P_i^{-1}(P_i(L))=\vert\vert_{i\in I} P_i(L),
\end{equation}
where the last equality holds according to Definition~\ref{product}. From the definitions of $\Sigma_c$ and $\Sigma_{uc}$, the above equation is equivalent to
\begin{equation}
L\Sigma_{uc}\cap L(G) \subseteq \vert\vert_{i\in I} P_i(L) = L,
\end{equation}
which enforces the global controllability of $L$.
\end{proof}

Based on the aforementioned properties of separately controllable languages, the following theorem is established to illustrate the correctness of the proposed mission planning scheme.

\begin{theorem}\label{sep cont correct}
The mission planning procedure shown in Fig.~\ref{caf} terminates and returns local mission supervisors $S_i$ as well as local mission plans $L^{mi}_i (i\in I)$ for each agent such that the collective behaviors of the cooperative multi-agent system $G$ achieve a separately controllable sublanguage of $L$, within a finite number of iterations.
\end{theorem}

\begin{proof}
When $L$ is separable, the local mission plans are given by $L^{mi}_i:=\sup C_i(P_i(L))$ and the correctness of the proposed mission planning procedure trivially holds.

When $L$ is not separable, we define
\begin{equation}\nonumber
S(L)=\{H\subseteq L \vert (L=\overline{L})\land (H=\vert\vert_{i\in I} P_i(H))\}
\end{equation}
as the set of prefix-closed and separable sublanguages of $L$. Under the assumption imposed by Remark 2, $S(L)\ne\emptyset$.

At each iteration step of the $L^*_{CV}$ algorithm, a counterexample $t$ is actually a word that belongs to $\vert\vert_{i\in I} L^{mi}_i - L$. The first step (\ref{co1}) of the counterexample-guided replanning thus eliminates all the observationally indistinguishable local mission behaviors with respect to $t$ from the local mission plan $L^{mi}_i$ and generates a temporal mission language $L_i^{temp}$; while the second step (\ref{co2}) in fact computes the supremal prefix-closed sublanguage \cite{kum} of $L_i^{temp} (i\in I)$. Thus,
\begin{equation}
L_i=\overline{L_i}\subseteq L^{temp}_i.
\end{equation}

Since during the iterative execution of the mission planning, $L_i$ is used as an updated mission specification for the $L^*_{LS}$ algorithm; thus, the updated mission plans of agent $G_i$, denoted as $\tilde L^{mi}_i (i\in I)$, after the generation of $t$, will be
\begin{equation}
\tilde L^{mi}_i=\sup C_i(L_i)\subseteq L_i \subseteq L^{temp}_i \subseteq L^{mi}_i
\end{equation}
which implies that the joint effort of the $L^*_{LS}$ and $L^*_{CV}$ algorithms generates the following ``monotonic" sequence of local mission plans for each agent $G_i$:
\begin{equation*}
\emptyset\subseteq \cdots \subseteq \tilde L^{mi}_i \subseteq L^{mi}_i\subseteq L^{mi}_{i,0}:=\sup C_i(P_i(L))\subseteq P_i(L).
\end{equation*}

When no more counterexamples are produced by the $L^*_{CV}$ algorithm, with slightly abusing the notations, we can write that
$$\tilde L :=\vert\vert_{i\in I} \tilde L^{mi}_i \subseteq L,$$
thus $\tilde L\in S(L)$. Furthermore, $\tilde L^{mi}_i$, generated by the $L^*_{LS}$ algorithm, is controllable with respect to $\Sigma_{i,uc}$ and $G_i$. Hence according to Definition~\ref{separate controllability}, $\vert\vert_{i\in I} \tilde L^{mi}_i$ forms a non-empty and separately controllable sublanguage of $L$, which solves the mission planning of the cooperative multi-agent system.
%During the execution, the $L_{CV}^*$ keeps presenting counterexamples $t$ such that $t\in K-S(K)$. At each iteration step, $K_i$ and $P_i(\overline{t})$ are both prefix-closed, hence $K_i-P_i(\overline{t})$ is prefix-closed; moreover, when no more counterexample $t$ is provided by the $L_{CV}^*$, one can eventually obtain a sublanguage, namely $\tilde K\in S(K)$. Facing the separable specification $\tilde K$, the automatic synthesis framework can then synthesize the local supervisors $\tilde S_i, i\in I$ to fulfill the supremal controllable sublanguage of $\tilde K$, which is also a sublanguage of $K$ and this language is a solution of Problem 1.
\end{proof}

\begin{remark}
We use $L^{mi}_i$ to denote the synthesized mission plans for agent $G_i (i\in I)$ hereafter.
\end{remark}

\section{Motion Planning of Cooperative Agents}\label{sec:verification}
\subsection{Automatic Generation of Motion Plans}
After the design of feasible local missions $L^{mi}_i (i\in I)$, the motion planning problem concentrates on finding a set of motion plans $L^{mo}_i$ associated with each $L^{mi}_i$ such that the integrated mission-motion plans $LP_i (i\in I)$ solve Problem~\ref{dcccp}. A local motion plan actually consists of two parts: a local motion plan $L^{mo}_i\subseteq V^*$ that enumerates all the regions visited by the agent, and a local door profile $D^{mo}_i\subseteq D^*$ associated with $L^{mo}_i$ that records all possible doors through which the agents shall pass. It is desired that the local motion plan is {\it adequate}, which is defined as follows.
\begin{definition}[Adequate Plans]\label{perm}
A local motion plan $L^{mo}_i$ for agent $G_i (i\in I)$ is said to be adequate if
\begin{enumerate}
\item $L^{mo}_i\models \pi_i(L^{mi}_i)$;
\item $L^{mo}_i\subseteq Run\left[L(G^m_i)\right]$.
%\item for any other local motion plan $\hat {L}^{mo}_i$ that satisfies 1) and 2), $L^{mo}_i\subseteq \hat {L}^{mo}_i$.
\end{enumerate}
\end{definition}

From Definition~\ref{perm}, synthesis of the local motion plan starts by exploiting the mission-motion integration relation $\pi_i (i\in I)$ in (5). The intuition behind Definition~\ref{perm} is that $L^{mo}_i$ shall obey restrictions imposed not only by $L^{mi}_i$, but by $G^m_i$ as well (adequacy).
%; secondly, among all the local motion plans that satisfy the adequacy, $L^{mo}_i$ contains the words of minimal lengths.

It is clear from (5) that the mission-motion integration mapping $\pi_i$ maps the prefix-closed language $L^{mi}_i$ over $\Sigma_i$ into a prefix-closed language $\pi_i(L^{mi}_i)$ over $V$. To efficiently compute an adequate motion plan for each agent (which is not known {\it a priori}), we employ a third modification of the $L^*$ learning algorithm, termed as the $L^*_{MP}$\footnote{The subscript ``MP" stands for ``motion planning".}  algorithm, to pursue an adequate local motion plan. Specifically, the Teacher of the $L^*_{MP}$ algorithm is designed to determine the following membership queries for agent $G_i$: for $i\in I$ and $t\in V^*$:
\begin{equation}\label{L_mp}
T_i(t) =
\begin{cases}
1,  & \mbox{if } \mathcal{DFA}(t)\models \pi_i(L^{mi}_i) \mbox{ is true, }\\
0, & \mbox{otherwise,}
\end{cases}
\end{equation}
where $\mathcal{DFA}(t)$ is defined similarly to the membership queries (\ref{L_cv}) for the $L^*_{CV}$ algorithm, except that the underlying DFA is defined over $V$. Besides justifying the membership queries (\ref{L_mp}), the Teacher of the $L^*_{MP}$ algorithm answers the conjecture $L^{mo}_i\subseteq Run\left[L(G^m_i)\right]$. If the conjecture is denied, a counterexample $t\in V^*$ is produced. Since $L^{mo}_i\subseteq Run\left[L(G^m_i)\right]$ is false, we know that $t$ witnesses a difference between $L^{mo}_i$ and $Run\left[L(G^m_i)\right]$; therefore, it is returned to the $L^*_{MP}$ algorithm to update the motion plan $L^{mo}_i$.

The correctness and termination properties of the $L^*_{MP}$ algorithm can be summarized as the following theorem.
\begin{theorem}
Given the local mission plan $L^{mi}_i$ and the mission-motion integration mapping $\pi_i$, the $L^*_{MP}$ algorithm terminates and correctly constructs an adequate local motion plan $L^{mo}_i (i\in I)$.
\end{theorem}
\begin{proof}
({\it Correctness}) The Teacher of the $L^*_{MP}$ algorithm is designed according to the two properties in Definition~\ref{perm}. When a counterexample is generated, the counterexample in fact represents a word in the symmetric difference between $\pi_i(L^{mi}_i)$ and $Run\left[L(G^m_i)\right]$; therefore, if no more counterexample is generated, the current learned DFA $L^{mo}_i$ (with slightly abusing the notations) satisfies both of the requirements of Definition~\ref{perm}, which guarantees the correctness.

({\it Termination}) At any iteration step, after a local motion plan DFA $L^{mo}_i$ is conjectured, the $L^*_{MP}$ algorithm reports whether or not $L^{mo}_i$ is adequate and terminates, or continues the construction of $L^{mo}_i$ by providing new counterexamples. By Theorem~\ref{L_ter}, the learning procedure eventually terminates at some iteration step $j\in \mathbb{N}$, at that time, the $L^*_{MP}$ algorithm produces an adequate $L^{mo}_i$.
\end{proof}

%As illustrated in Algorithm 2 in details, $PMP-SYN$ extracts a feasible and permissive motion plan from the motion diagram $G^m_i$ for each agent $G_i$, $i\in I$. We starts with $\pi_i(L^{mi}_i) \subseteq V^*$ that specifies all the regions the agent $G_i$ should visit sequentially in order to accomplish missions in $L^{mi}_i$ (line 1). We check the feasibility of $\pi_i(L^{mi}_i)$, i.e., whether or not $\pi_i(L^{mi}_i)\subseteq Run\left[L(G^m_i)\right]$, if the set inclusion holds, then $\pi_i(L^{mi}_i)$
%Therefore, it can be accepted by a DFA $\pi_i(L^{mi}_i)$ over $V$, with slightly abusing the notations.

Furthermore, the door profile associated with the synthesized local motion plan $L^{mo}_i$ is obtained as
\begin{equation}
D^{mo}_i={\rm Word}(L^{mo}_i),
\end{equation}
where the operator $\rm Word$ computes all the corresponding words $s\in D^*$ by treating words in $L^{mo}_i$ as runs of $G^m_i$, i.e., ${\rm Word}(L^{mo}_i)=\{s\in L(G^m_i)\vert Run(s)\in L^{mo}_i\}$.
\subsection{Counterexample-guided Replanning in Uncertain Environments}
We aim to implement the integrated plans $\{LP_i\}_{i\in I}$ in the real environment $\mathcal{E}$ rather than the nominal environment characterized by $\mathcal{E}_0$, which is used for generating motion plans and door profiles. As mentioned in Section~\ref{sec:formulation}.A, feasible transitions among adjacent regions in $\mathcal{E}_0$ may become infeasible in $\mathcal{E}$ since some doors that are supposed to be open are closed. During the online coordination of the cooperative multi-agent systems, the underlying agents may detect the real door mapping $F^\mathcal{E}_D$ and the real region transition diagram $\delta^\mathcal{E}_i$ through sensing capabilities and communication with each other. Here we assume that the motion capacities $$G^\mathcal{E}_i=(V,D,v_{i,0},\delta^\mathcal{E}_i,V)$$
of agent $G_i (i\in I)$, under the restriction of $\mathcal{E}$, are also captured by a trim DFA. Based on the online-acquired knowledge of the environment, we develop a new algorithm to address the environment uncertainty, whose outputs include both the feasible integrated plan $LP^\mathcal{E}_i$ and the associated door profile $D^m_i$ for agent $G_i (i\in I)$ moving and performing missions in $\mathcal{E}$.

\begin{algorithm}[t]
\caption{Implementation of $LP_i$ in $\mathcal{E}$}
\begin{algorithmic}[1]
\REQUIRE Local integrated plan $LP_i$, local motion model $G^m_i$ and local door profile $D^{mo}_i$
\ENSURE Local implementable integrated plan $LP^\mathcal{E}_i$
\STATE Initialization: $LP^\mathcal{E}_i=LP_i$, $D^m_i=D^{mo}_i$
\STATE Order the words in $LP_i$: $LP_i=\bigcup_{k=1}^K \overline{LP_i^k}$
\STATE Construct $D^{mo}_i=\bigcup_{k=1}^K \overline{D^k_i}$: $D^k_i={\rm Word}[P_V(LP^k_i)]$
\FOR {all $k\in \{1,2,\ldots,K\}$}
\STATE Let $LP_i^k=LP_i^k(0)LP_i^k(1)\ldots LP_i^k(m)$
\FOR {all $l\in \{1,2,\ldots,m\}$}
\IF { $(\exists l)(\exists d): [LP_i^k(l+1)=\delta^\mathcal{E}_i(LP_i^k(l),d)]$ but $[d\not\in F_D^\mathcal{E}(LP_i^k(l),LP_i^k(l+1))]$ }
\STATE Find $d'\in D^k_i: LP_i^k(l+1)=\delta^m_i(LP_i^k(l),d')$ and $d'\in F_D^\mathcal{E}(LP_i^k(l),LP_i^k(l+1))$
\STATE Remove the words in $D^k_i$ that contains $d$
\STATE $LP^\mathcal{E}_i=LP_i$
\ELSE
\STATE Check if there exists $l$ such that $LP_i^k(l+1)\ne\delta^\mathcal{E}_i(LP_i^k(l),d)$ for all $d\in D$
\IF {such $l$ does exist}
\STATE Find $v_1,v_2,\ldots,v_p:$$(LP_i^k(l),v_1),\ldots,$$(v_p,LP_i^k(l+1))\in\longrightarrow_\mathcal{E}$
\STATE $LP_i^k(l)LP^k_i(l+1)=LP_i^k(l)v_1\ldots v_pLP^k_i(l+1)$
\STATE Update $LP^\mathcal{E}_i$ according to $LP^k_i$
\STATE Update $D^m_i$ accordingly
\ENDIF
\ENDIF
\ENDFOR
\ENDFOR
\RETURN $LP^\mathcal{E}_i$ and $D^m_i$
\end{algorithmic}
\end{algorithm}

As detailed in Algorithm 2, the implementation of $LP_i$ in the real environment $\mathcal{E}$ starts with the synthesized local plans $LP_i$ and corresponding door profiles $D^{mo}_i$. For $i\in I$, Algorithm 2 first enumerates (line 2) $LP_i$ as a collection of $K$ prefix-closed words over $V\cup \Sigma_i$, i.e., $LP_i=\bigcup_{k=1}^K \overline{LP_i^k}$, while if $LP_i^k$ admits a cycle such that $LP_i^k=uv^*$ for some $u\in\Sigma^*$ and $v\in\Sigma^*-\{\epsilon\}$, we replace $LP_i^k$ by $uv$. Next for each $LP^k_i$, we proceed to the investigation of two types of environment uncertainties. On the one hand (lines 7-10), if two adjacent regions in $LP_i$ are connected by multiple doors whereas some of the doors are closed, we require the agent to use alternative (redundant) doors to accomplish the motion plan and hence the integrated local plan $LP_i$ remains the same, while the door profile $D^m_i$ is formed by discarding all the words in $D^{mo}_i$ that contain the symbols of closed doors. On the other hand (lines 12-17), when two consecutive regions $v$ and $v'$ that are supposed to be visited by agent $G_i$ are not connected by any doors in $\mathcal{E}$, the assumption that $G^\mathcal{E}_i$ is a trim DFA inspires us to replace the motion transition $vv'$ by a sequence of ``intermediate" regions $v_1,v_2,\ldots,v_p$ such that
$$v_{i,0}\ldots vv_1v_2\ldots v_pv' \in Run[L(G^\mathcal{E}_i)];$$
afterwards, we replace all the transition $vv'$ in $LP_i$ by $vv_1v_2\ldots v_pv'$ to construct $LP^\mathcal{E}_i$, and the door profile is updated accordingly.

\section{Case Study: A Robotic Coordination Example}\label{sec:CoSMoP}

We examine the effectiveness of our proposed learning-based synthesis framework with a multi-robot coordination case study, whose scenario is shown in Fig.~\ref{scenario}.

\begin{figure}[b]
\centering
\includegraphics[width=0.35\textwidth]{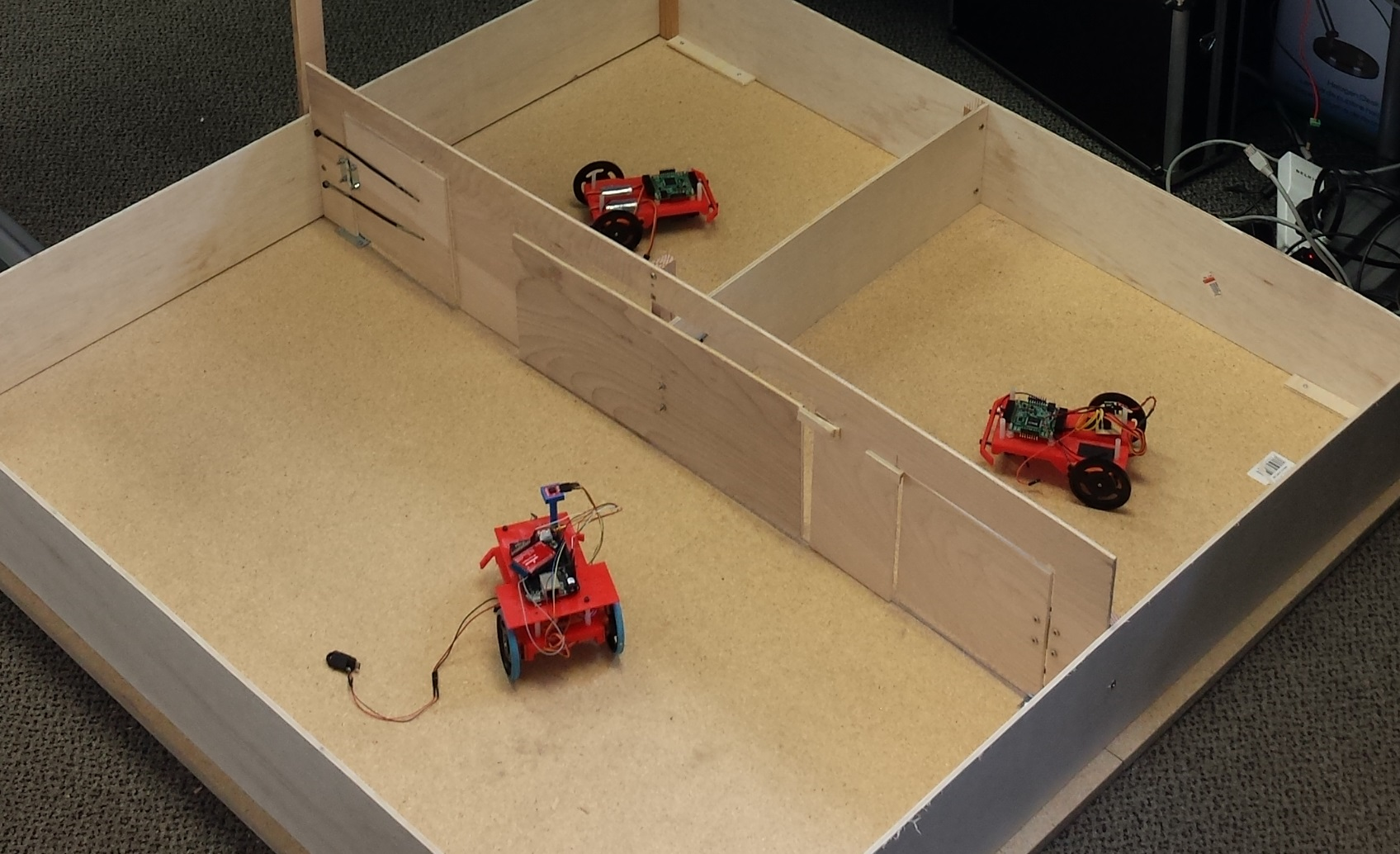}
\caption{The scenario of multi-robot coordination.}
\label{scenario}
\vspace{-5mm}
\end{figure}

\subsection{Description of the Multi-robot Coordination Scenario}
As depicted in Fig.~\ref{scenario}, the cooperative robotic team consists of three robots, namely $G_1$, $G_2$ and $G_3$, all of which have identical communication and self-localization capabilities. In addition, $G_2$ is assumed to be equipped with rescue and fire-fighting capabilities. We assume that all doors are equipped with a spring and are kept closed whenever there is no force to keep them open. Initially, the three robots are positioned in Room 1. Room 2 and Room 3 are accessible using the one-way door $D_2$, or the two-way doors $D_1$ and $D_3$, respectively. The nominal environment of the robots is depicted in Fig.~\ref{scenario1}
\begin{figure}[h]
\centering
\includegraphics[width=0.2\textwidth]{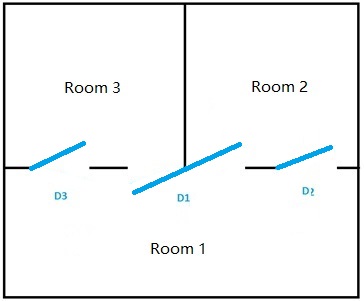}
\caption{The nominal environment $\mathcal{E}_0$.}
\label{scenario1}
\vspace{-7mm}
\end{figure}

From (3), the nominal environment $\mathcal{E}_0$ is characterized by:
\begin{itemize}
\item $V=\{R_1,R_2,R_3\}$, where $R_j$ stands for Room $j (j=1,2,3)$;
\item $\longrightarrow_{\mathcal{E}_0}=\{(R_1,R_2),(R_1,R_3),(R_2,R_1),(R_3,R_1)\}$;
\item $D=\{D_1^l,D_1^r,D_2,D_3\}$, where $D_1^l$ and $D_1^r$ represent the left and right half of $D_1$, respectively;
\item $F_D(R_1,R_2)=\{D_1^r,D_2\}$, $F_D(R_2,R_1)=\{D_1^r\}$, $F_D(R_1,R_3)=\{D_1^l,D_3\}$, and $F_D(R_3,R_1)=\{D_1^l,D_3\}$.
\end{itemize}

For $i\in I=\{1,2,3\}$, the motion DFA $G^m_i=(V,D,v_{i,0},\delta^m_i,V)$ for robot $G_i$ is then given by: $v_{i,0}=R_1$, $\delta^m_i(R_1,D_1^r)=\delta^m_i(R_1,D_2)=R_2$, $\delta^m_i(R_1,D_3)=\delta^m_i(R_1,D_1^l)=R_3$, $\delta^m_i(R_3,D_3)=\delta^m_i(R_3,D_1^l)=R_1$. Thus, the DFA representation of $G^m_i$ is given as follows:

\begin{figure}[h]
	\centering
	\begin{tikzpicture}[shorten >=1pt,node distance=2.1cm,on grid,auto, bend angle=20, thick,scale=0.65, every node/.style={transform shape}]
	\node[state,initial] (s_0)   {$R_1$};
	\node[state] (s_1) [right of = s_0] {$R_2$};
	\node[state] (s_2) [left of = s_0] {$R_3$};
	\path[->]
	(s_0) edge [bend left] node [pos=0.5, sloped, above]{$D_1^r,D_2$} (s_1)
    (s_0) edge [bend left] node [pos=0.5, sloped, below]{$D_1^l,D_3$} (s_2)
	(s_1) edge [bend left] node [pos=0.5, sloped, below]{$D_1^r$} (s_0)
	(s_2) edge [bend left] node [pos=0.5, sloped, above]{$D_1^l,D_3$} (s_0);
	\end{tikzpicture}
	\caption{$G^m_i$ for robot $G_i (i=1,2,3)$.}
	\vspace{-2.5mm}
\end{figure}

The coordination mission assigned to this multi-robot system requires that after a fire-extinguishing alarm from $R_2$ is triggered, $G_2$ needs to go to $R_2$ and come back immediately to $R_1$ through the two-way door $D_1$ that can only be opened with the cooperation of $G_1$ and $G_3$. To open $D_1$ efficiently, either $G_1$ or $G_3$ goes to $R_3$ from the two-way door $D_3$, while the other stays in $R_1$ to synchronously open $D_1$ and waits for $G_2$ to return to $R_1$. Afterwards, both $G_1$ and $G_3$ move backwards to close $D_1$ and all three robots assemble in $R_1$ for next service request.
%Intuitively, the specifications for the multi-robot coordination are twofold.
%\begin{itemize}
%\item {\it Request-response:} Robot $R_2$ should respond to the request by entering Room 2 through $D_2$ and returning to Room 1 through $D_1$ when it accomplishes its task;
%\item {\it Coordination:} $D_1$ should be opened jointly by $R_1$ and $R_3$.
%\end{itemize}

\subsection{Automata-based Characterization of Multi-robot Systems}
The global mission events of the multi-robot system are given by $\Sigma=\bigcup_{i=1}^3\Sigma_i$, and  are listed in Table~1.

\begin{table}[H]
  \begin{center}
    \caption{Mission Events of the Robots.}
    \label{list_of_events}
    \begin{tabular}{ll} \toprule
      \multicolumn{1}{l}{Event} &  {Explanation} \\ \midrule
      $h_i$ & Robot $G_i (i=1,2,3)$ receives the mission request.\\
      $F$ & Robot $G_2$ extinguishes the fire.  \\
      %$G_ionD_1$ & Robot $G_i (i=1,3)$ localizes itself at the door $D_1$,. \\
      %$R_itok$ & Robot $R_i$ heads for Room $k$, $i,k=1,2,3$. \\
      $G_iinR_j$ & Robot $G_i$ stays at Room $R_j (i,j=1,2,3)$. \\
      $Open$ & Open $D_1$.  \\
      $Close$ & Close $D_1$. \\
      $D_1open$ & The status that $D_1$ is opened by robots. \\
      $D_1close$ & The status that $D_1$ is closed by $G_1$ and $G_3$. \\
      $r$ & All the robots assemble in Room 1.\\\bottomrule
    \end{tabular}
  \end{center}
  \vspace{-7mm}
\end{table}

We assume that the prior knowledge of all local mission event sets $\Sigma_i (i=1,2,3)$ is accessible to the supervisor:
\begin{equation}\nonumber
\begin{split}
\Sigma_1=\{& h_1,Open,Close,G_2inR_1,G_1inR_3,D_1close,\\
& D_1open,G_1inR_1,r\},\\
\Sigma_2=\{& h_2,F,D_1open,G_2inR_1,r\},\\
\Sigma_3=\{& h_3,G_3inR_3,Open,Close,D_1open,G_2inR_1,\\
& D_1close,G_3inR_1,r\}.
\end{split}
\end{equation}
Furthermore, for the purpose of supervisory control, we assume that $\Sigma_{1,uc}=\{h_1,G_2inR_1\}$, $\Sigma_{2,uc}=\{h_2,D_1open\}$ and $\Sigma_{3,uc}=\{h_3,G_2inR_1\}$.

%Two regular language specifications, denoted by $K_s, K_d\subseteq \Sigma^*$ (subscripts $s$ and $d$ stand for ``service" and "door", respectively), that correspond to the aforementioned performance requirements, are shown in Fig.~4, represented by DFA, where $h_1||h_3$ stands for the synchronous product of single-event traces $\{h_1\}$ and $\{h_3\}$.
As discussed in the previous subsection, the following two regular languages $L^{spe}_1$ and $L^{spe}_2$ are introduced to capture the two requirements for the multi-robot system, respectively:
\begin{equation*}
\begin{split}
L^{spe}_1=&\overline{(h_2FD_1openG_2inR_1r)^*} \\
L^{spe}_2=&\overline{((h_1h_3+h_3h_1)[(G_1inR_1G_3inR_3+G_3inR_3G_1inR_1)}\\
& \overline{+(G_3inR_1G_1inR_3+G_1inR_3 G_3inR_1)]} \\
& \overline{OpenD_1openG_2inR_1CloseD_1close \ r)^*}
\end{split}
\end{equation*}

The overall global mission specification for the multi-agent system is then $L=L^{spe}_1\vert\vert L^{spe}_2$. Our design objective is to synthesize valid local integrated plans $LP_i=(L^{mi}_i,L^{mo}_i)$ for each robot $G_i (i=1,2,3)$ such that $\vert\vert_{i=1}^3 LP_i \models L$ holds in the real environment $\mathcal{E}$.

\subsection{Synthesis of Local Plans}
%The automatic synthesis framework presented in Section III is detailed in Fig. 5:
%
%\begin{figure}[H]
%\centering
%\includegraphics[width=200pt]{images/caf}
%\caption{Learning-based automatic synthesis framework.}
%\label{caf1}
%\end{figure}

We first follow the framework shown in Fig.~\ref{caf} to synthesize local mission plans $L^{mi}_i$ and local mission supervisors $S_i$ for robot $G_i (i=1,2,3)$. The initial local missions $L_i$ are obtained by $L_i=P_i(L) (i=1,2,3)$, whose DFA representations are shown in Figures \ref{fig:l13} and \ref{fig:l2}, respectively.
\begin{figure}[H]
	\centering
	\subfloat[$L_1$ for $G_1$]{\label{fig:ks}	
		\begin{tikzpicture}[shorten >=1pt,node distance=2.1cm,on grid,auto, bend angle=20, thick,scale=0.45, every node/.style={transform shape}]
		\node[state,initial] (s_0)   {};
		\node[state] (s_1) [right=of s_0] {};
		\node[state] (s_2) [right=of s_1] {};
		\node[state] (s_3) [right=of s_2] {};
		\node[state] (s_4) [below=of s_3] {};
		\node[state] (s_5) [left=of s_4] {};
		\node[state] (s_6) [left=of s_5] {};
		\node[state] (s_7) [left=of s_6] {};
        \node[state] (s_8) [above=of s_1] {};
        \node[state] (s_9) [right=of s_8] {};
        \node[state] (s_10) [right=of s_9] {};
        \node[state] (s_11) [above=of s_10] {};
        \node[state] (s_12) [left=of s_11] {};
        \node[state] (s_13) [left=of s_12] {};
        \node[state] (s_14) [left=of s_13] {};
        %\node[state] (s_15) [below=of s_14] {};
		%\node[state] (s_8) [left=of s_7] {};
	    %\node[state] (s_9) [below=of s_1] {};
%        \node[state] (s_10) [right=of s_9] {};
		\path[->]
		(s_0) edge node [pos=0.5, sloped, above]{$h_1$} (s_1)
		%(s_1) edge node [pos=0.5, sloped, above]{$R_1 to D_1$} (s_2)
        (s_1) edge node [pos=0.5, sloped, above]{$G_1inR_1$} (s_2)
        %(s_1) edge [bend right] node [pos=0.5, sloped, below]{$R_1inG_3$} (s_2)
		%(s_2) edge node [pos=0.5, sloped, above]{$R_1 on D_1$} (s_3)
		(s_2) edge node [pos=0.5, sloped, above]{$Open$} (s_3)	
		(s_3) edge node [pos=0.5, sloped, above]{$D_1open$} (s_4)
		(s_4) edge node [pos=0.5, sloped, above]{$G_2inR_1$} (s_5)
		(s_5) edge node [pos=0.5, sloped, above]{$Close$} (s_6)
		(s_6) edge node [pos=0.5, sloped, above]{$D_1closed$} (s_7)
		(s_7) edge node [pos=0.5, sloped, above]{$r$} (s_0)
        (s_1) edge node [pos=0.5, sloped, above]{$G_1inR_3$} (s_8)
        (s_8) edge node [pos=0.5, sloped, above]{$Open$} (s_9)	
		(s_9) edge node [pos=0.5, sloped, above]{$D_1open$} (s_10)
		(s_10) edge node [pos=0.5, sloped, above]{$G_2inR_1$} (s_11)
		(s_11) edge node [pos=0.5, sloped, above]{$Close$} (s_12)
		(s_12) edge node [pos=0.5, sloped, above]{$D_1close$} (s_13)
        (s_13) edge node [pos=0.5, sloped, above]{$G_1inR_1$} (s_14)
        (s_14) edge node [pos=0.5, sloped, above]{$r$} (s_0);
        %(s_9) edge node [pos=0.5, sloped, above]{$R_1 in 3$} (s_10)
%        (s_10) edge node [pos=0.5, sloped, above]{$R_1 to D_1$} (s_2);   	
		\end{tikzpicture}}
	\subfloat[$L_3$ for $G_3$]{\begin{tikzpicture}[shorten >=1pt,node distance=2.1cm,on grid,auto, bend angle=20, thick,scale=0.45, every node/.style={transform shape}]
		\node[state,initial] (s_0)   {};
		\node[state] (s_1) [right=of s_0] {};
		\node[state] (s_2) [right=of s_1] {};
		\node[state] (s_3) [right=of s_2] {};
		\node[state] (s_4) [below=of s_3] {};
		\node[state] (s_5) [left=of s_4] {};
		\node[state] (s_6) [left=of s_5] {};
		\node[state] (s_7) [left=of s_6] {};
        \node[state] (s_8) [above=of s_1] {};
        \node[state] (s_9) [right=of s_8] {};
        \node[state] (s_10) [right=of s_9] {};
        \node[state] (s_11) [above=of s_10] {};
        \node[state] (s_12) [left=of s_11] {};
        \node[state] (s_13) [left=of s_12] {};
        \node[state] (s_14) [left=of s_13] {};
        %\node[state] (s_15) [below=of s_14] {};
		%\node[state] (s_8) [left=of s_7] {};
	    %\node[state] (s_9) [below=of s_1] {};
%        \node[state] (s_10) [right=of s_9] {};
		\path[->]
		(s_0) edge node [pos=0.5, sloped, above]{$h_1$} (s_1)
		%(s_1) edge node [pos=0.5, sloped, above]{$R_1 to D_1$} (s_2)
        (s_1) edge node [pos=0.5, sloped, above]{$G_3inR_1$} (s_2)
        %(s_1) edge [bend right] node [pos=0.5, sloped, below]{$R_1inG_3$} (s_2)
		%(s_2) edge node [pos=0.5, sloped, above]{$R_1 on D_1$} (s_3)
		(s_2) edge node [pos=0.5, sloped, above]{$Open$} (s_3)	
		(s_3) edge node [pos=0.5, sloped, above]{$D_1open$} (s_4)
		(s_4) edge node [pos=0.5, sloped, above]{$G_2inR_1$} (s_5)
		(s_5) edge node [pos=0.5, sloped, above]{$Close$} (s_6)
		(s_6) edge node [pos=0.5, sloped, above]{$D_1closed$} (s_7)
		(s_7) edge node [pos=0.5, sloped, above]{$r$} (s_0)
        (s_1) edge node [pos=0.5, sloped, above]{$G_3inR_3$} (s_8)
        (s_8) edge node [pos=0.5, sloped, above]{$Open$} (s_9)	
		(s_9) edge node [pos=0.5, sloped, above]{$D_1open$} (s_10)
		(s_10) edge node [pos=0.5, sloped, above]{$G_2inR_1$} (s_11)
		(s_11) edge node [pos=0.5, sloped, above]{$Close$} (s_12)
		(s_12) edge node [pos=0.5, sloped, above]{$D_1close$} (s_13)
        (s_13) edge node [pos=0.5, sloped, above]{$G_3inR_1$} (s_14)
        (s_14) edge node [pos=0.5, sloped, above]{$r$} (s_0);  	
		\end{tikzpicture}}
	\caption{Local mission specifications $L_1$ and $L_3$}
	\label{fig:l13}
\vspace{-2.5mm}
\end{figure}

\begin{figure}[H]
	\centering
	\begin{tikzpicture}[shorten >=1pt,node distance=2.1cm,on grid,auto, bend angle=20, thick,scale=0.45, every node/.style={transform shape}]
	\node[state,initial] (s_0)   {};
	\node[state] (s_1) [right=of s_0] {};
	\node[state] (s_2) [right=of s_1] {};
	\node[state] (s_3) [below=of s_2] {};
	\node[state] (s_4) [left=of s_3] {};
	%\node[state] (s_5) [left=of s_4] {};
%	\node[state] (s_6) [left=of s_5] {};
	\path[->]
	(s_0) edge node [pos=0.5, sloped, above]{$h_2$} (s_1)
	(s_1) edge node [pos=0.5, sloped, above]{$F$} (s_2)
	(s_2) edge node [pos=0.5, sloped, above]{$D_1open$} (s_3)
	(s_3) edge node [pos=0.5, sloped, above]{$G_1inR_1$} (s_4)	
	%(s_4) edge node [pos=0.5, sloped, above]{$R_2 to 1$} (s_5)
%	(s_5) edge node [pos=0.5, sloped, above]{$R_2 in 1$} (s_6)
	(s_4) edge node [pos=0.5, sloped, below]{$r$} (s_0);	
	\end{tikzpicture}
	\caption{Local mission specification $L_2$.}
	\label{fig:l2}
\vspace{-2.5mm}
\end{figure}

Next, we aim to apply the $L_{LS}^*$ algorithm for the synthesis of the local mission supervisor for robot $G_i$ to fulfill $L_i (i=1,2,3)$. For instance, a candidate supervisor $S_2$ demonstrated in Fig.~\ref{fig:CS2} is synthesized to achieve $\sup C_2(L_2)$ for robot $G_2$, which is identical to $L_2$. Local mission supervisors $S_1$ and $S_3$ can be developed in a similar manner, and it turns out that $\sup C_1(L_1)=L_1$ and $\sup C_3(L_3)=L_3$.
\begin{figure}[H]
	\centering
	%\begin{tikzpicture}[shorten >=1pt,node distance=2.1cm,on grid,auto, bend angle=20, thick,scale=0.55, every node/.style={transform shape}]
	\begin{tikzpicture}[shorten >=1pt,node distance=2.1cm,on grid,auto, bend angle=20, thick,scale=0.45, every node/.style={transform shape}]
	\node[state,initial] (s_0)   {};
	\node[state] (s_1) [right=of s_0] {};
	\node[state] (s_2) [right=of s_1] {};
	\node[state] (s_3) [below=of s_2] {};
	\node[state] (s_4) [left=of s_3] {};
	%\node[state] (s_5) [left=of s_4] {};
%	\node[state] (s_6) [left=of s_5] {};
	\path[->]
	(s_0) edge node [pos=0.5, sloped, above]{$h_2$} (s_1)
	(s_1) edge node [pos=0.5, sloped, above]{$F$} (s_2)
	(s_2) edge node [pos=0.5, sloped, above]{$D_1open$} (s_3)
	(s_3) edge node [pos=0.5, sloped, above]{$G_1inR_1$} (s_4)	
	%(s_4) edge node [pos=0.5, sloped, above]{$R_2 to 1$} (s_5)
%	(s_5) edge node [pos=0.5, sloped, above]{$R_2 in 1$} (s_6)
	(s_4) edge node [pos=0.5, sloped, below]{$r$} (s_0);	
	\end{tikzpicture}
	\caption{A candidate local mission supervisor $S_2$ for robot $G_2$.}
	\label{fig:CS2}
\vspace{-2.5mm}
\end{figure}

Afterwards, we present $M_i=S_i\vert\vert G_i$ to the $L^*_{CV}$ algorithm, and a counterexample that indicates a violation of $L$ emerges. We examine this counterexample and it turns out that when one of the robots $G_1$ or $G_3$ stays in $R_1$, the other robot must stay in $R_3$. By adding this counterexample back to the $L^*_{CV}$ algorithm, we update the local specifications according to Section~\ref{sec:synthesis}.E, the $L^*_{LS}$ algorithm derives the new local mission supervisors as shown in Figures \ref{fig:S13} and \ref{fig:S2}, respectively.
%\begin{figure}[H]
%	\centering
%		\begin{tikzpicture}[shorten >=1pt,node distance=2.1cm,on grid,auto, bend angle=20, thick,scale=0.55, every node/.style={transform shape}]
%		\node[state,initial] (s_0)   {};
%		\node[state] (s_1) [right=of s_0] {};
%		%\node[state] (s_2) [right=of s_1] {};
%		\node[state] (s_3) [right=of s_1] {};
%		\node[state] (s_4) [right=of s_3] {};
%		\node[state] (s_5) [below=of s_4] {};
%		\node[state] (s_6) [left=of s_5] {};
%		\node[state] (s_7) [left=of s_6] {};
%		\node[state] (s_8) [left=of s_7] {};	
%		\path[->]
%		(s_0) edge node [pos=0.5, sloped, above]{$h_1$} (s_1)
%		(s_1) edge node [pos=0.5, sloped, above]{$G_1inR_1$} (s_3)
%		%(s_2) edge node [pos=0.5, sloped, above]{$R_1 on D_1$} (s_3)
%		(s_3) edge node [pos=0.5, sloped, above]{$Open$} (s_4)	
%		(s_4) edge node [pos=0.5, sloped, above]{$D_1open$} (s_5)
%		(s_5) edge node [pos=0.5, sloped, above]{$G_2inR_1$} (s_6)
%		(s_6) edge node [pos=0.5, sloped, above]{$Close$} (s_7)
%		(s_7) edge node [pos=0.5, sloped, above]{$D_1closed$} (s_8)
%		(s_8) edge node [pos=0.5, sloped, above]{$r$} (s_0);	
%		\end{tikzpicture}
%
%		\caption{The local mission supervisor $S_1$.}
%		\label{fig:S1}
%\end{figure}

\begin{figure}[H]
\centering
\subfloat[The local mission supervisor $S_1$]{\label{fig:ks}	
		\begin{tikzpicture}[shorten >=1pt,node distance=2.1cm,on grid,auto, bend angle=20, thick,scale=0.42, every node/.style={transform shape}]
		\node[state,initial] (s_0)   {};
		\node[state] (s_1) [right=of s_0] {};
		%\node[state] (s_2) [right=of s_1] {};
		\node[state] (s_3) [right=of s_1] {};
		\node[state] (s_4) [right=of s_3] {};
		\node[state] (s_5) [below=of s_4] {};
		\node[state] (s_6) [left=of s_5] {};
		\node[state] (s_7) [left=of s_6] {};
		\node[state] (s_8) [left=of s_7] {};	
		\path[->]
		(s_0) edge node [pos=0.5, sloped, above]{$h_1$} (s_1)
		(s_1) edge node [pos=0.5, sloped, above]{$G_1inR_1$} (s_3)
		%(s_2) edge node [pos=0.5, sloped, above]{$R_1 on D_1$} (s_3)
		(s_3) edge node [pos=0.5, sloped, above]{$Open$} (s_4)	
		(s_4) edge node [pos=0.5, sloped, above]{$D_1open$} (s_5)
		(s_5) edge node [pos=0.5, sloped, above]{$G_2inR_1$} (s_6)
		(s_6) edge node [pos=0.5, sloped, above]{$Close$} (s_7)
		(s_7) edge node [pos=0.5, sloped, above]{$D_1close$} (s_8)
		(s_8) edge node [pos=0.5, sloped, above]{$r$} (s_0);	
		\end{tikzpicture}}
	\subfloat[The local mission supervisor $S_3$]{\begin{tikzpicture}[shorten >=1pt,node distance=2.1cm,on grid,auto, bend angle=20, thick,scale=0.42, every node/.style={transform shape}]
\node[state,initial] (s_0)   {};
\node[state] (s_1) [right=of s_0] {};
\node[state] (s_2) [right=of s_1] {};
%\node[state] (s_3) [right=of s_2] {};
%\node[state] (s_4) [right=of s_3] {};
\node[state] (s_5) [right=of s_1] {};
\node[state] (s_6) [right=of s_5] {};
\node[state] (s_7) [right=of s_6] {};
\node[state] (s_8) [below=of s_7] {};
\node[state] (s_9) [left=of s_8] {};
\node[state] (s_10)[left=of s_9] {};
\node[state] (s_11)[left=of s_10] {};
\path[->]
(s_0) edge node [pos=0.5, sloped, above]{$h_3$} (s_1)
(s_1) edge node [pos=0.5, sloped, above]{$G_3inR_3$} (s_5)
%(s_2) edge node [pos=0.5, sloped, above]{$R_3 in 3$} (s_3)
%(s_3) edge node [pos=0.5, sloped, above]{$R_3 to D_1$} (s_4)	
%(s_4) edge node [pos=0.5, sloped, above]{$R_3 on D_1$} (s_5)
(s_5) edge node [pos=0.5, sloped, above]{$Open$} (s_6)
(s_6) edge node [pos=0.5, sloped, above]{$D_1open$} (s_7)
(s_7) edge node [pos=0.5, sloped, above]{$G_2inR_1$} (s_8)
(s_8) edge node [pos=0.5, sloped, above]{$Close$} (s_9)
(s_9) edge node [pos=0.5, sloped, above]{$D_1close$} (s_10)
(s_10) edge node [pos=0.5, sloped, above]{$G_3inR_1$} (s_11)
(s_11) edge node [pos=0.5, sloped, above]{$r$} (s_0);		
\end{tikzpicture}}
\caption{Local mission supervisors for $G_1$ and $G_3$.}
\label{fig:S13}
\vspace{-2.5mm}
\end{figure}
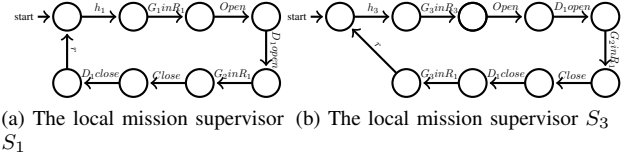

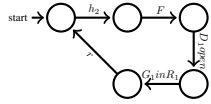
\begin{figure}[H]
	\centering
	\begin{tikzpicture}[shorten >=1pt,node distance=2.1cm,on grid,auto, bend angle=20, thick,scale=0.42, every node/.style={transform shape}]
	\node[state,initial] (s_0)   {};
	\node[state] (s_1) [right=of s_0] {};
	\node[state] (s_2) [right=of s_1] {};
	\node[state] (s_3) [below=of s_2] {};
	\node[state] (s_4) [left=of s_3] {};
	%\node[state] (s_5) [left=of s_4] {};
%	\node[state] (s_6) [left=of s_5] {};
	\path[->]
	(s_0) edge node [pos=0.5, sloped, above]{$h_2$} (s_1)
	(s_1) edge node [pos=0.5, sloped, above]{$F$} (s_2)
	(s_2) edge node [pos=0.5, sloped, above]{$D_1open$} (s_3)
	(s_3) edge node [pos=0.5, sloped, above]{$G_1inR_1$} (s_4)	
	%(s_4) edge node [pos=0.5, sloped, above]{$R_2 to 1$} (s_5)
%	(s_5) edge node [pos=0.5, sloped, above]{$R_2 in 1$} (s_6)
	(s_4) edge node [pos=0.5, sloped, below]{$r$} (s_0);	
	\end{tikzpicture}
	\caption{The local mission supervisor $S_2$.}
	\label{fig:S2}
\vspace{-2.5mm}
\end{figure}

%\begin{figure}[H]
%	\centering
%\begin{tikzpicture}[shorten >=1pt,node distance=2.1cm,on grid,auto, bend angle=20, thick,scale=0.55, every node/.style={transform shape}]
%\node[state,initial] (s_0)   {};
%\node[state] (s_1) [right=of s_0] {};
%\node[state] (s_2) [right=of s_1] {};
%%\node[state] (s_3) [right=of s_2] {};
%%\node[state] (s_4) [right=of s_3] {};
%\node[state] (s_5) [right=of s_1] {};
%\node[state] (s_6) [right=of s_5] {};
%\node[state] (s_7) [below=of s_6] {};
%\node[state] (s_8) [left=of s_7] {};
%\node[state] (s_9) [left=of s_8] {};
%\node[state] (s_10)[left=of s_9] {};	
%\path[->]
%(s_0) edge node [pos=0.5, sloped, above]{$h_3$} (s_1)
%(s_1) edge node [pos=0.5, sloped, above]{$G_3inR_3$} (s_5)
%%(s_2) edge node [pos=0.5, sloped, above]{$R_3 in 3$} (s_3)
%%(s_3) edge node [pos=0.5, sloped, above]{$R_3 to D_1$} (s_4)	
%%(s_4) edge node [pos=0.5, sloped, above]{$R_3 on D_1$} (s_5)
%(s_5) edge node [pos=0.5, sloped, above]{$Open$} (s_6)
%(s_6) edge node [pos=0.5, sloped, above]{$D_1open$} (s_7)
%(s_7) edge node [pos=0.5, sloped, above]{$G_2inR_1$} (s_8)
%(s_8) edge node [pos=0.5, sloped, above]{$Close$} (s_9)
%(s_9) edge node [pos=0.5, sloped, above]{$D_1close$} (s_10)
%(s_10) edge node [pos=0.5, sloped, above]{$r$} (s_0);		
%\end{tikzpicture}
%	\caption{Local supervisor $S_3$.}
%	\label{fig:S3}
%\end{figure}

Intuitively, the updated local mission plans $S_1$ and $S_3$ explicitly specify that robot $G_1$ stays in $R_1$, while robot $G_3$ enters $R_3$ such that $D_1$ can be opened jointly. We insert the new $M_i=S_i \vert\vert G_i$ to the $L^*_{CV}$ algorithm to justify whether the collective behaviors of the local supervisors cooperatively satisfy $L$; and in this time the $L^*_{CV}$ algorithm reports no more counterexample, which implies that $S_i (i=1,2,3)$ drives $G_i$ to accomplish an appropriate local mission. Therefore, Fig.~\ref{fig:S13} and Fig.~\ref{fig:S2} are also DFA representations of local mission plans $L^{mi}_i (i=1,2,3)$.

Finally, we apply the $L^*_{MP}$ algorithm stated in Section~\ref{sec:verification} to construct local motion plan $L^{mo}_i$ that is associated with $L^{mi}_i (i=1,2,3)$, respectively. Towards this end, the local motion plans are established in terms of DFAs:

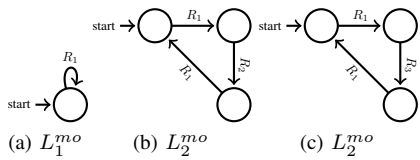
\begin{figure}[H]
	\centering
	\subfloat[$L^{mo}_1$]{\begin{tikzpicture}[shorten >=1pt,node distance=2.1cm,on grid,auto, bend angle=20, thick,scale=0.5, every node/.style={transform shape}]
	\node[state,initial] (s_0)   {};
	\path[->]
	(s_0) edge [loop above] node [pos=0.5, sloped, above]{$R_1$} (s_0);
	\end{tikzpicture}}
    \subfloat[$L^{mo}_2$]{\begin{tikzpicture}[shorten >=1pt,node distance=2.1cm,on grid,auto, bend angle=20, thick,scale=0.5, every node/.style={transform shape}]
	\node[state,initial] (s_0)   {};
	\node[state] (s_1) [right=of s_0] {};
    \node[state] (s_2) [below=of s_1] {};
	\path[->]
	(s_0) edge node [pos=0.5, sloped, above]{$R_1$} (s_1)
	(s_1) edge node [bend left] [pos=0.5, sloped, above]{$R_2$} (s_2)
	(s_2) edge node [bend left] [pos=0.5, sloped, below]{$R_1$} (s_0);
	\end{tikzpicture}}
\subfloat[$L^{mo}_2$]{\begin{tikzpicture}[shorten >=1pt,node distance=2.1cm,on grid,auto, bend angle=20, thick,scale=0.5, every node/.style={transform shape}]
	\node[state,initial] (s_0)   {};
	\node[state] (s_1) [right=of s_0] {};
\node[state] (s_2) [below=of s_1] {};
	\path[->]
	(s_0) edge node [pos=0.5, sloped, above]{$R_1$} (s_1)
	(s_1) edge node [pos=0.5, sloped, above]{$R_3$} (s_2)
	(s_2) edge node [pos=0.5, sloped, below]{$R_1$} (s_0);
	\end{tikzpicture}}
	\caption{Local motion plans $L^{mo}_i$ for robot $G_i (i=1,2,3)$.}
\label{fig:mp}
\vspace{-2.5mm}
\end{figure}

Based on the local motion plans obtained as shown in Fig.~\ref{fig:mp}, we can construct the integrated local plan $LP_i$ for each robot $G_i$. Such construction relies on the local mission-motion integration mapping $\pi_i (i=1,2,3)$, which is defined as follows:
\begin{enumerate}
\item Robot $G_1$:
\begin{equation*}
\begin{split}
&\pi_1(h_1)=\pi_1(Open)=\pi_1(Close)=\pi_1(G_2inR_1)\\
&=\pi_1(D_1open)=\pi_1(D_1close)=\pi_1(r)=R_1, \\
&\pi_1(G_1inR_3)=R_3.
\end{split}
\end{equation*}
\item Robot $G_2$:
\begin{equation*}
\begin{split}
&\pi_2(h_2)=\pi_2(G_2inR_1)=\pi_2(r)=R_1,\\
&\pi_2(F)=\pi_2(D_1open)=R_2.
\end{split}
\end{equation*}
\item Robot $G_3$:
\begin{equation*}
\begin{split}
&\pi_3(Open)=\pi_3(Close)=\pi_3(G_2inR_1)\\
&=\pi_3(D_1open)=\pi_3(D_1close)=\pi_1(r)=R_3, \\
&\pi_3(G_1inR_3)=\pi_3(h_3)=\pi_3(r)=R_1.
\end{split}
\end{equation*}
\end{enumerate}

Fig.~\ref{fig:LP1}, Fig.~\ref{fig:LP2} and Fig.~\ref{fig:LP3} depict the integrated local plans $LP_1$, $LP_2$ and $LP_3$, respectively.

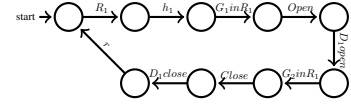
\begin{figure}[H]
	\centering
	\begin{tikzpicture}[shorten >=1pt,node distance=2.1cm,on grid,auto, bend angle=20, thick,scale=0.42, every node/.style={transform shape}]
	\node[state,initial] (s_0)   {};
		\node[state] (s_1) [right=of s_0] {};
		\node[state] (s_2) [right=of s_1] {};
		\node[state] (s_3) [right=of s_2] {};
		\node[state] (s_4) [right=of s_3] {};
		\node[state] (s_5) [below=of s_4] {};
		\node[state] (s_6) [left=of s_5] {};
		\node[state] (s_7) [left=of s_6] {};
		\node[state] (s_8) [left=of s_7] {};	
		\path[->]
		(s_0) edge node [pos=0.5, sloped, above]{$R_1$} (s_1)
        (s_1) edge node [pos=0.5, sloped, above]{$h_1$} (s_2)
		(s_2) edge node [pos=0.5, sloped, above]{$G_1inR_1$} (s_3)
		%(s_2) edge node [pos=0.5, sloped, above]{$R_1 on D_1$} (s_3)
		(s_3) edge node [pos=0.5, sloped, above]{$Open$} (s_4)	
		(s_4) edge node [pos=0.5, sloped, above]{$D_1open$} (s_5)
		(s_5) edge node [pos=0.5, sloped, above]{$G_2inR_1$} (s_6)
		(s_6) edge node [pos=0.5, sloped, above]{$Close$} (s_7)
		(s_7) edge node [pos=0.5, sloped, above]{$D_1close$} (s_8)
		(s_8) edge node [pos=0.5, sloped, above]{$r$} (s_0);	
	\end{tikzpicture}
	\caption{The integrated local plan $LP_1$.}
	\label{fig:LP1}
\vspace{-2.5mm}
\end{figure}

 \begin{figure}[H]
	\centering
	\begin{tikzpicture}[shorten >=1pt,node distance=2.1cm,on grid,auto, bend angle=20, thick,scale=0.42, every node/.style={transform shape}]
	\node[state,initial] (s_0)   {};
	\node[state] (s_1) [right=of s_0] {};
	\node[state] (s_2) [right=of s_1] {};
	\node[state] (s_3) [right=of s_2] {};
	\node[state] (s_4) [below=of s_3] {};
	\node[state] (s_5) [left=of s_4] {};
	\node[state] (s_6) [left=of s_5] {};
    \node[state] (s_7) [left=of s_6] {};
	\path[->]
	(s_0) edge node [pos=0.5, sloped, above]{$R_1$} (s_1)
    (s_1) edge node [pos=0.5, sloped, above]{$h_2$} (s_2)
	(s_2) edge node [pos=0.5, sloped, above]{$R_2$} (s_3)
	(s_3) edge node [pos=0.5, sloped, above]{$F$} (s_4)
	(s_4) edge node [pos=0.5, sloped, above]{$D_1open$} (s_5)	
	(s_5) edge node [pos=0.5, sloped, above]{$R_1$} (s_6)
	(s_6) edge node [pos=0.5, sloped, above]{$G_2inR_1$} (s_7)
	(s_7) edge node [pos=0.5, sloped, below]{$r$} (s_0);	
	\end{tikzpicture}
	\caption{The integrated local plan $LP_2$.}
	\label{fig:LP2}
\vspace{-2.5mm}
\end{figure}
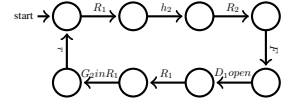

 \begin{figure}[H]
	\centering
	\begin{tikzpicture}[shorten >=1pt,node distance=2.1cm,on grid,auto, bend angle=20, thick,scale=0.42, every node/.style={transform shape}]
	\node[state,initial] (s_0)   {};
	\node[state] (s_1) [right=of s_0] {};
	\node[state] (s_2) [right=of s_1] {};
	\node[state] (s_3) [right=of s_2] {};
	\node[state] (s_4) [right=of s_3] {};
	\node[state] (s_5) [right=of s_4] {};
	\node[state] (s_6) [below=of s_5] {};
    \node[state] (s_7) [left=of s_6] {};
    \node[state] (s_8) [left=of s_7] {};
    \node[state] (s_9) [left=of s_8] {};
    \node[state] (s_10) [left=of s_9] {};
    \node[state] (s_11) [left=of s_10] {};
    %\node[state] (s_6) [left=of s_5] {};
	\path[->]
	(s_0) edge node [pos=0.5, sloped, above]{$R_1$} (s_1)
	(s_1) edge node [pos=0.5, sloped, above]{$h_3$} (s_2)
	(s_2) edge node [pos=0.5, sloped, above]{$R_3$} (s_3)
    (s_3) edge node [pos=0.5, sloped, above]{$G_3inR_3$} (s_4)
    (s_4) edge node [pos=0.5, sloped, above]{$Open$} (s_5)
	(s_5) edge node [pos=0.5, sloped, above]{$D_1open$} (s_6)	
	(s_6) edge node [pos=0.5, sloped, above]{$G_2inR_1$} (s_7)
	(s_7) edge node [pos=0.5, sloped, above]{$Close$} (s_8)
    (s_8) edge node [pos=0.5, sloped, above]{$D_1close$} (s_9)
    (s_9) edge node [pos=0.5, sloped, above]{$R_1$} (s_10)
    (s_10) edge node [pos=0.5, sloped, above]{$G_3inR_1$} (s_11)
	(s_11) edge node [pos=0.5, sloped, below]{$r$} (s_0);	
	\end{tikzpicture}
	\caption{The integrated local plan $LP_3$.}
	\label{fig:LP3}
\vspace{-2.5mm}
\end{figure}
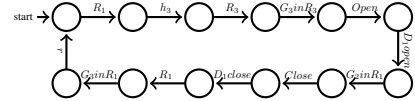

It is worth pointing out that, all the $LP_i$'s are designed with respect to $\mathcal{E}_0$, and when implementing $LP_i$'s in the practical environment $\mathcal{E}$, the robots may encounter with uncertainties. In this case, the robots shall automatically react to the uncertain environment by online replanning their motion plans. For example, in case the door $D_3$ is broken and cannot be opened when $G_3$ needs to return to $R_1$, $G_3$ shall then return to $R_1$ through $D_1$.
%=====================================================================================
\section{Experimental Results}\label{sec:experimental_setup}
\subsection{Experiment Setup}

We now proceed to the experimental validation of the theoretical results. The proposed supervisors, implemented as MATLAB StateFlow machines, are tested and evaluated on a team composed by three robots.

The experimental setup for testing the derived supervisors is composed by three main elements, described respectively as follows.

\subsubsection{Localization System}

The localization system is constructed on the basis of the Aruco library \cite{aruco} and provides the status of each robot (position and orientation) and each door (open or close) with a frame-rate of $\sim20$Hz. Each robot is equipped with a unique marker. The localization software detects the position and the orientation of each marker, providing the position data through a UDP socket. The information concerning the position of the robots and the status of each door is encoded using the JSON format.
Furthermore, each robot is able to request its position or another robot's position.

The localization software is based on a multi-threading architecture in order to minimize the delay on the transmission of the position information:
\begin{itemize}
\item The main thread implements a UDP server waiting for information requests from the robots. The response contains the robot id, its position and orientation and the status of all the doors.

For example, the position of the robot $G_1$ along with the status of the three doors is represented as follows:
 $\{G1:[+0460.86,+0113.43,+0079.15,0,0,1]\}$.
\item The second thread computes the position of each marker and updates the list of detected markers along with their position and orientation. The camera used on the localization system has a resolution of 640x480 and a frame-rate of $\sim30$ frames per second.
\end{itemize}

Figure~\ref{localization_system} shows the output of the localization system.
\begin{figure}[h]
\centering
\includegraphics[width=200pt]{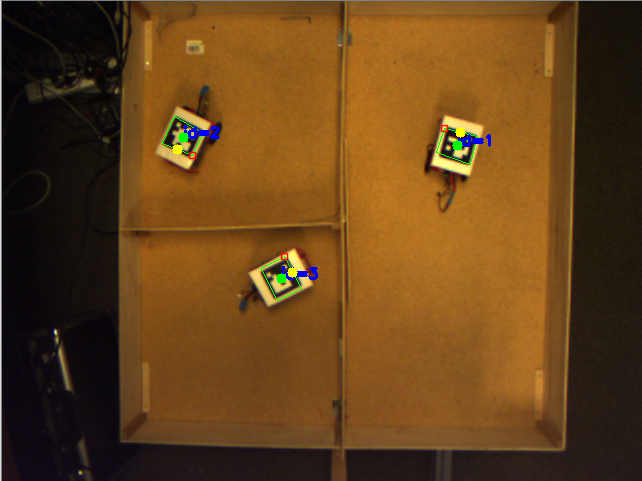}
\caption{View of the system from the camera used for the localization system with markers for position and orientation estimation. The localization system provides visual feedback with augmented-reality capabilities: each marker is overlapped with basic information (ID of the marker, orientation). }
\label{localization_system}
\vspace{-2.5mm}
\end{figure}

\subsubsection{Local supervisors}
Each of the three supervisors is implemented as a StateFlow machine and manages a pre-defined robot of the team.

The logic implemented follows the theoretical approach described in Section \ref{sec:CoSMoP}. The transition between the  possible states is based on the occurrence of one of the pre-defined events, e.g., the opening of a door. The list of possible events is represented in Table \ref{list_of_events}.

The communication between the robot and the corresponding StateFlow machine is realized using a wireless module. The TCP protocol is chosen to provide a reliable exchange of messages.

As an illustrative example, Figures~\ref{state_flow_high_level} and~\ref{supervisor_robot2} show the information transition between the robot and the local plan, as well as the StateFlow implementation of the local plan $LP_2$ for robot $G_2$, respectively.
\begin{figure}[h]
\centering
\includegraphics[width=0.4\textwidth]{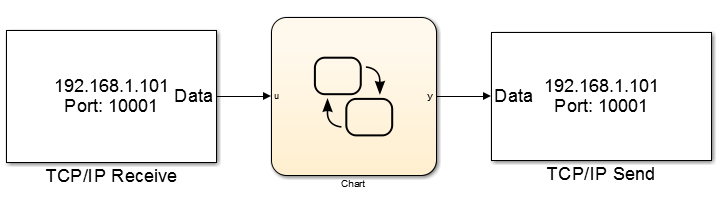}
\caption{High-level structure of the local plan $LP_2$ implemented as Stateflow machine with TCP/IP Receive/Send communication blocks. The TCP/IP Receive module is based on a blocking mechanism: The simulink model makes an iteration each time a new packet is received. Once the results is sent back to the robot, the state flow returns to its blocking state listening for a new packet.}
\label{state_flow_high_level}
\vspace{-2.5mm}
\end{figure}
\begin{figure}[h]
\centering
\includegraphics[width=0.4\textwidth]{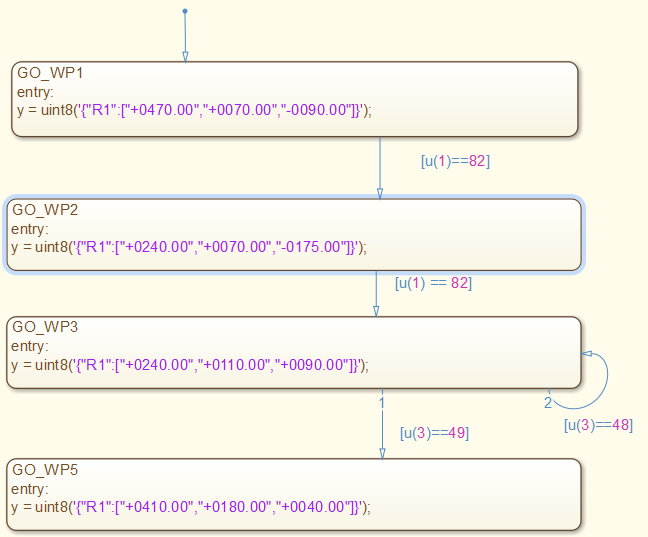}
\caption{StateFlow Representation of local supervisor $LP_2$ using ASCII encoding of the events. Since StateFlow deals only with numbers, the information is codified using the ASCII representation. The data sent back to the robot are then converted to characters.}
\label{supervisor_robot2}
\vspace{-2.5mm}
\end{figure}

\subsubsection{Ground Robots}
The hardware of each robot is designed as a multi-layered architecture.
\begin{itemize}
\item The low-level system (LL board), based on xMOS processors \cite{goncalo1}, is responsible to interact with the sensors and the actuators (Fig.~\ref{low_level_label}).
\begin{figure}[h]
\centering
\includegraphics[width=0.3\textwidth]{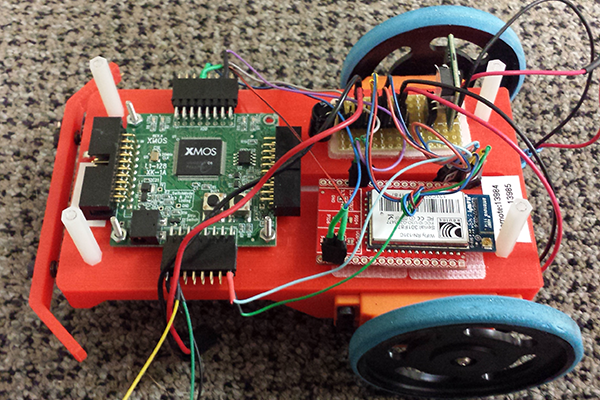}
\caption{Low-level system. The xMOS device is responsible to interact with sensors. The UART over IP module shown in the figure is used for requesting the robot's position from the localization server. The power supply is realized using a switching regulator.}
\label{low_level_label}
\vspace{-2.5mm}
\end{figure}
Thanks to the hard real-time performance of the xMOS devices, the low-level system provides deterministic real-time access to all the sensors and actuators connected to the processor.
\item The high-level system (HL board) is constructed on the basis of the BeagleBone board and is responsible to interact with the corresponding local plan (Fig.~\ref{high_level}). The main purpose of the high-level board is to provide a tool for developing high-level control systems, delegating the low-level real-time board for the interaction with the sensors and the actuators.
\begin{figure}[h]
\centering
\includegraphics[width=0.3\textwidth]{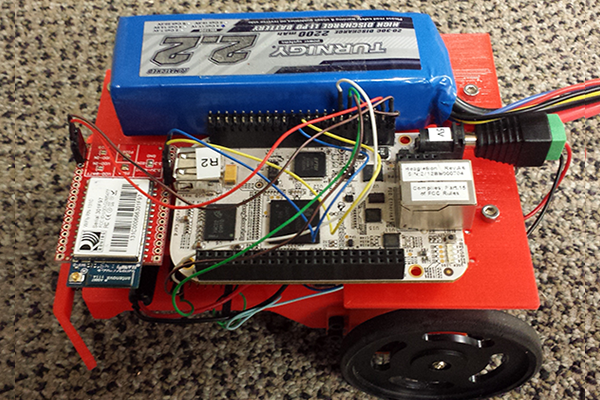}
\caption{High-level system. The Beaglebone board device is responsible to interact with State Flow. The UART over IP module shown in the figure is used for requesting the command to be executed from the localization supervisor. The power supply is based on a 3-cells LiPo battery.}
\label{high_level}
\vspace{-2.5mm}
\end{figure}
\end{itemize}

Taking advantage of the localization system (Fig.~\ref{localization_system}), each robot is equipped with a unique marker (Fig.~\ref{marker}) that is convenient to be recognized by the localization system. The coordinates of the four corners of the marker are used to estimate the position (in terms of pixel coordinates) of the center of the marker, along with its orientation.

\begin{figure}[h]
\centering
\includegraphics[width=0.3\textwidth]{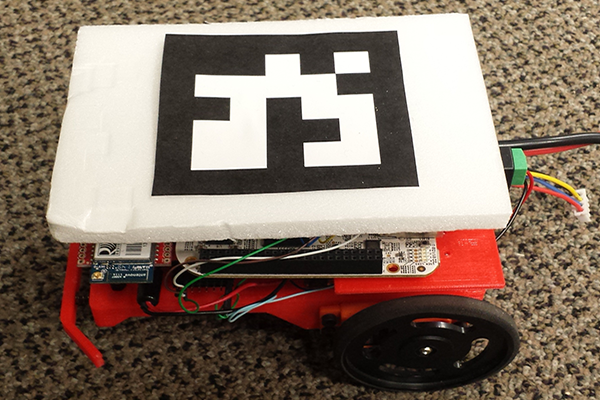}
\caption{Marker mounted on top of the robot. The localization system is able to detect the position of the center of the marker and its orientation with a frame-rate of $\sim20Hz$. The Aruco library, used in the localization system, provides the possibility to detect up to 1024 different markers.}
\label{marker}
\vspace{-2.5mm}
\end{figure}

\subsection{Demonstration Results}
The flow chart in Fig.~\ref{flow_chart} describes the interaction between the
local plan (represented in DFAs) and the corresponding ground robot.
\label{experimental_results}
\begin{figure}[h]
\centering
\includegraphics[width=0.35\textwidth]{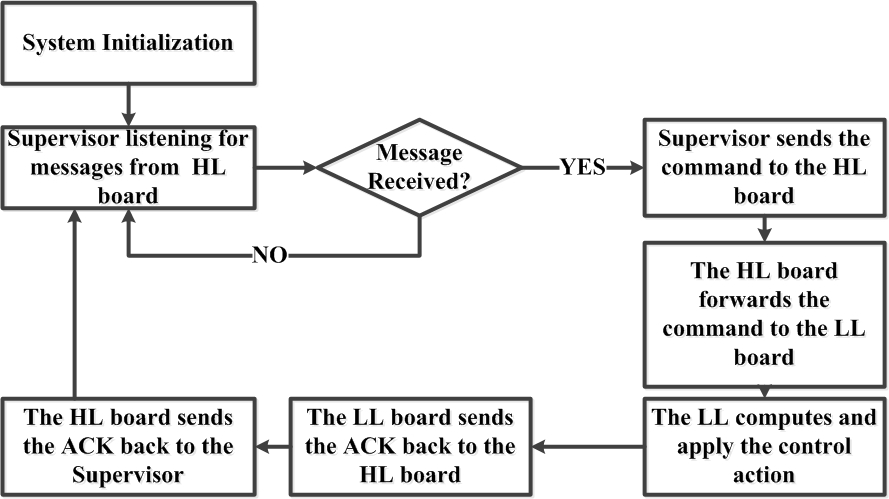}
\caption{High-level flow chart of the system. The local plan (DFA) is able to communicate with the High Level board of the robot. The HL board will then translate the received command in control action for the motors. The use of acknowledge messages ensures the correctness of the messages transmitted.}
\label{flow_chart}
\vspace{-2.5mm}
\end{figure}

After the reception of the START signal, the High-level board of the agent initializes the communication system and sends the first request to its local plan DFA. The
local plan DFA responds by sending the command back to the HL board. The HL board then reads data from the sensors and computes control actions interacting with the LL board. After reaching a desired position, the HL board sends an ACK (acknowledge) message back to the supervisor. Once the local plan supervisor receives the ACK message, it replies by sending the next action to the HL board of the agent. For example, we implement the local plan $LP_2$ (Fig.~\ref{fig:LP2}) for robot $G_2$ in the real experiment environment. The path followed by robot $G_2$, as defined by the local motion plan $L^{mo}_2$, is shown in Fig.~\ref{path_robot2}.

\begin{figure}[h]
\centering
\includegraphics[width=0.27\textwidth]{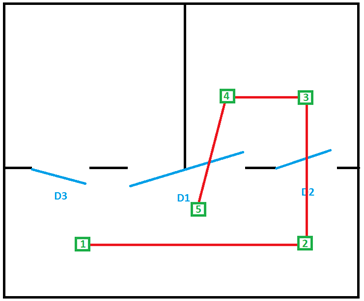}
\caption{Path planned for robot $G_2$. The path is described as a series of waypoints to be reached by the robot. For some waypoints, additional conditions must be satisfied. For example, robot 2 cannot reach waypoint 5 until $D_1$ is closed.}
\label{path_robot2}
\vspace{-2.5mm}
\end{figure}

In Figure~\ref{actual_path_robot_2}, we record the waypoints visited by robot $G_2$ and plot the actual path. The cloud of points corresponds to the instants in which the robot
$G_2$ reached a pre-defined waypoint and waited for the next command sent from
the local supervisor. The position of the robot is expressed in pixel
coordinates. The differences between the planned path and the real path are
due to hardware limitations, in particular to the difference between the
localization rate and the update rate of the servo motors. This problem can be
overcome by integrating additional sensors such as Inertial Measuring Unit
(IMU) or encoders in order to provide a higher data rate for position
estimation, as well as an estimation of the robot's position even when the
marker is not detected. A sensor fusion algorithm is currently under development to overcome both limitations.

\begin{figure}[h]
\centering
\includegraphics[width=0.38\textwidth]{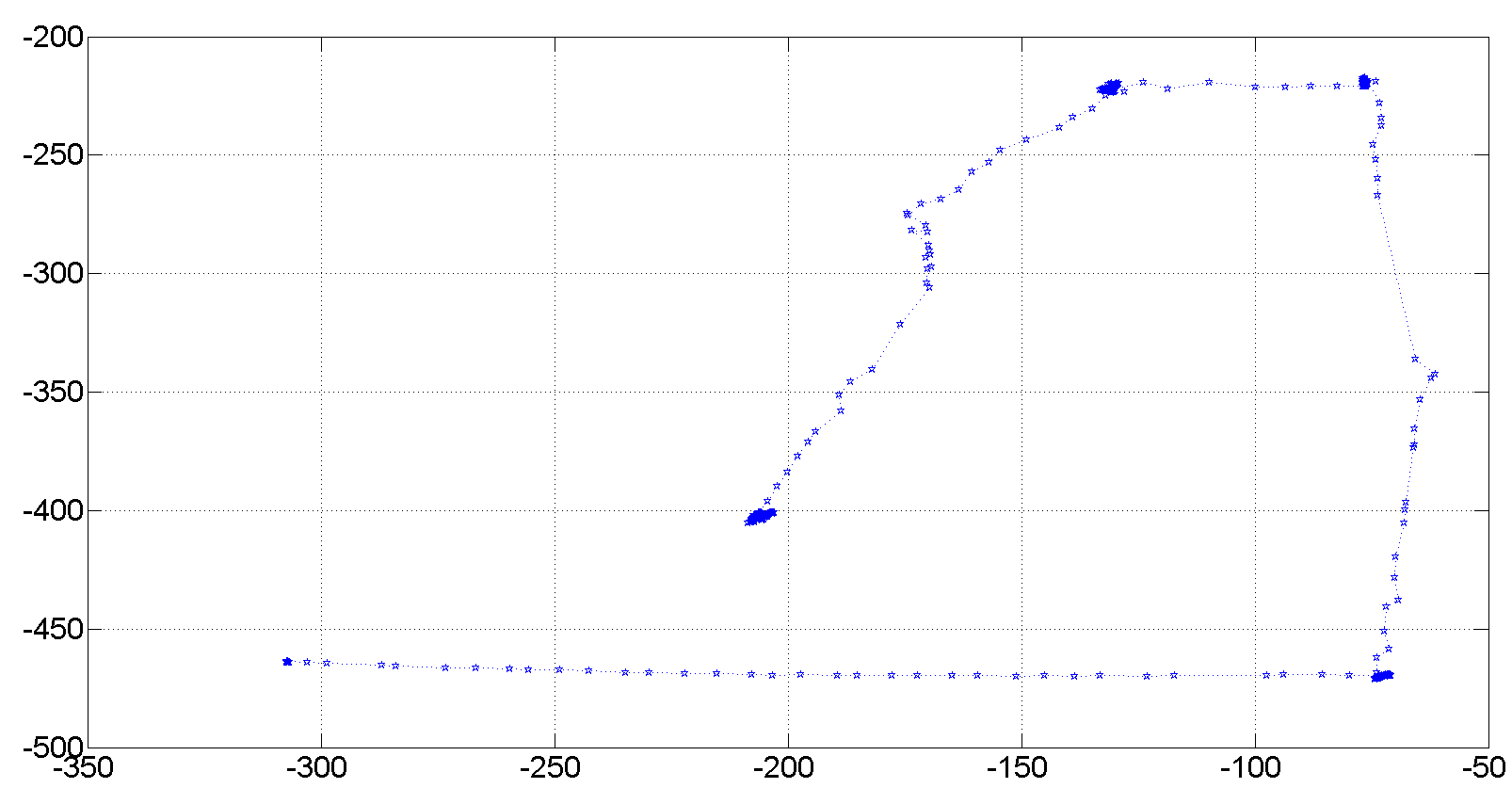}
\caption{Path followed by robot $G_2$. The position is updated with a frame-rate of 20 Hz. The clouds of points correspond to the instants in which the robot reached a waypoint and stopped waiting for the next command.}
\label{actual_path_robot_2}
\vspace{-5mm}
\end{figure}

\section{Conclusion and Future Work}
\label{sec:conclusion}
In this paper, based on a formal design approach, we present a learning-based automatic synthesis framework to solve the coordination and control co-design problem of cooperative multi-agent systems. Starting from an automata-based characterization of each agent's dynamics as well as the shared environment in which all the agents perform their mission and motion behaviors, our proposed framework solves the mission and motion planning sub-problems in a hierarchical architecture. We develop three modified $L^*$ algorithms to synthesize the local mission supervisors, to check the joint efforts of the mission executions via compositional verification techniques, and to synthesize motion plans corresponding to the mission plans. In addition, we also present another algorithm that adaptively replans the mission and motions so that environment uncertainties can be resolved. It has been shown that the proposed synthesis framework can accomplish the global task even if the agents' models are not given {\it a priori}. Computational and software tools are also developed to incorporate automatic supervisor synthesis with inter-robot communication in an effort to implement the framework. A demonstration example of multi-robot coordination to achieve a request-response team mission is presented to justify the effectiveness of the proposed framework. Future work includes combining the proposed learning-based automatic synthesis framework with robustness requirements to deal with possible faults.

%  \newpage
%\fi

%\end{IEEEbiography}


%\newpage


\end{document}